\documentclass{compositio}
\usepackage{amsmath,amssymb,amsfonts}
\usepackage{graphicx}
\usepackage{tikz}

\usepackage[titletoc]{appendix}
 	\usepackage{pst-all}
	\usepackage{pstricks-add}

\def\PIV{\mathrm{P}_{\mathrm{IV}}}

\numberwithin{equation}{section}

\newtheorem{theorem}{Theorem}[section]
\newtheorem{proposition}[theorem]{Proposition}

\newtheorem{lemma}[theorem]{Lemma}
\newtheorem{remark}[theorem]{Remark}

\date{}

\begin{document}

\title[Asymptotic behaviour of the fourth Painlev\'e transcendents]{Asymptotic behaviour of the fourth Painlev\'e transcendents in the space of initial values}

\author{Nalini Joshi}
\email{nalini.joshi@sydney.edu.au}
\address{School of Mathematics and Statistics F07, University of Sydney, NSW 2006, Australia}

\author{Milena Radnovi\'c}
\email{milena.radnovic@sydney.edu.au}
\address{School of Mathematics and Statistics F07, University of Sydney, NSW 2006, Australia
\newline
Mathematical Institute SANU, Belgrade, Serbia}


\thanks{
The research reported in this paper was 
supported by grant no. FL120100094 from the Australian Research Council.
The work of M.R.~was partially supported by Project
no.~174020: \emph{Geometry and Topology of Manifolds, Classical Mechanics and Integrable
Dynamical Systems} of the Serbian Ministry of Education, Science and Technological Development.
\newline
M.R.~thanks Viktoria Heu for useful discussions.
}

\begin{abstract}
We study the asymptotic behaviour of solutions of the fourth Pain\-lev\'e equation as the independent variable goes to infinity in its space of (complex) initial values, which is a generalisation of phase space described by Okamoto.
We show that the limit set of each solution is compact and connected and, moreover, that any non-special solution has an infinite number of poles and infinite number of zeroes.
\end{abstract}
\maketitle

\section{Introduction}\label{sec:intro}
We study the dynamics of solutions of the fourth Painlev\'e equation 
\begin{equation}\label{eq:PIV}
\PIV\ : \ \frac{d^2y}{dx^2}=\frac{1}{2y}\left(\frac{dy}{dx}\right)^2+\frac32y^3+4xy^2+2(x^2-\alpha)y+\frac{\beta}{y},
\end{equation}
where $y=y(x)$ is a function of  $x\in\mathbb{C}$, and $\alpha,\beta$ complex constants, in the singular limit as $|x|\to\infty$ in the space of initial values, a generalisation of phase space first constructed in  \cite{Okamoto1979}. In this paper, we prove that each non-rational transcendental solution of $\PIV$ has infinitely many zeroes and poles in $\mathbb C$ (see Theorem \ref{th:poleszeroes}).

We start by transforming $\PIV$ to new coordinates that make the study of the limit $|x|\to\infty$ more explicit. The proof contains three ingredients: (i) the resolution of singularities of the Painlev\'e vector field in the space of initial values; (ii) an analytic study of the flow of the Painlev\'e vector field close to the exceptional lines in the resolved space; and (iii) construction of the complex limit set of each solution.  Using (i) and (ii), we prove that a certain set, called the infinity set, acts as a repeller of the Painlev\'e flow in Okamoto's space as $|x|\to\infty$ (see Theorem \ref{th:asymptotics}). Based on the estimates in the proof of this result, we show that the limit set of solutions is non-empty, compact, connected, and invariant under the flow of the associated autonomous system (see Theorem \ref{th:K}). Then by showing that the flow intersects infinitely often with the last three exceptional lines in the space of initial values, we prove Theorem \ref{th:poleszeroes}.
Earlier papers by one of us provided analogous results for the first and second Painlev\'e equations (\cite{DJ2011,HJ2014}). 

The fourth Painlev\'e equation has been studied from various perspectives:  see e.g., \cite{Okamoto1986,BCH1993,NO1997,NY1999,mat1999,Clark2003JMP,GJP2005,GJP2006,Stoy2014,NP2014}. However, the study of asymptotic behaviours in the limit $|x|\to\infty$ for $x\in{\mathbb C}$ appears to be incomplete in the literature. In this paper, we provide global information about the solutions' limiting behaviours in the complex plane in this singular limit. 

In Section \ref{sec:okamoto} we decribe the construction of Okamoto's space of initial values for Equation \eqref{eq:PIV}.
Basic steps of the resolution procedure are given there, but details of the calculations appear in Appendix \ref{sec:res-piv}.
Section \ref{sec:special} is devoted to the special solutions of the fourth Painlev\'e equation and their relation with singular curves in the elliptic pencil underlying the autonomous system.
Section \ref{sec:infinity} contains the results on asymptotic behaviour of the solutions and contains the proof of Theorem \ref{th:asymptotics}, while Section \ref{sec:limitset} provides information about limit sets and contains the proofs of Theorems \ref{th:K} and \ref{th:poleszeroes}.

\section{Space of Initial Values of $\PIV$}\label{sec:okamoto}
The fourth Painlev\'e equation (\ref{eq:PIV})
is equivalent to the following system:
\begin{equation}\label{eq:sistem}
\begin{aligned}
\frac{dy_1}{dx}&=-y_1(y_1+2y_2+2x)-2\alpha_1,
\\
\frac{dy_2}{dx}&=y_2(2y_1+y_2+2x)-2\alpha_2,
\end{aligned}
\end{equation}
with $y=y_1$, $\alpha=1-\alpha_1-2\alpha_2$, $\beta=-2\alpha_1^2$.
System (\ref{eq:sistem}) is Hamiltonian with the following Hamiltonian function:
\begin{equation}
H(x,y_1,y_2)=-y_1y_2(y_1+y_2+2x)+2\alpha_2y_1-2\alpha_1y_2,
\end{equation}
that is,  \eqref{eq:sistem} is equivalent to Hamilton's equations of motion
\begin{equation*}
	\frac{dy_1}{dx}=\frac{\partial H}{\partial y_2},\quad
	\frac{dy_2}{dx}= -\frac{\partial H}{\partial y_1} .
\end{equation*}

The asymptotic behaviour of the Painlev\'e transcendents was first studied by Boutroux \cite{Boutroux1913,Boutroux1914}.
There, for the first Painlev\'e equation, he made certain change of variables in order to make the asymptotic behaviours more explicit.
In the same spirit, we make the following change of variables for (\ref{eq:sistem}):
$$
y_1=xu,\quad y_2=xv,\quad z=\frac{x^2}2 
$$
which transforms the system (\ref{eq:sistem}) to
\begin{equation}\label{eq:sistemz}
\begin{aligned}
u'&=-u(u+2v+2)-\frac{\alpha_1}z-\frac{u}{2z},
\\
v'&=v(2u+v+2)-\frac{\alpha_2}z-\frac{v}{2z}.
\end{aligned}
\end{equation}
Here and later in this paper, primes denote differentiation with respect to $z$.

For each $z\neq0$, and each $(u_0,v_0)\in\mathbb{C}^2$, there is a unique solution of (\ref{eq:sistemz}) satisfying the initial conditions
$u(z_0)=u_0$, $v(z_0)=v_0$.
Since the solutions are meromorphic and therefore will become unbounded in neighbourhoods of movable poles, it is natural to consider the solutions as maps from $\mathbb{C}$ to $\mathbb{CP}^2$.
However,  for any given $z_0\neq0$, infinitely many  solutions may pass through certain points in $\mathbb{CP}^2$.
Such points will be called \emph{base points} in this paper.

To resolve the flow through such points, we need to construct the \emph{space of initial conditions} (see \cite{Gerard}), where the graph of each solution will represent a separate leaf of the foliation.
The spaces of initial conditions for all six Painlev\'e equations were constructed in \cite{Okamoto1979}.
The solutions are separated by resolving (i.e., blowing up) the base points.

In this paper, we explicitly construct such a resolution  of the system (\ref{eq:sistemz}).
The  details of the calculation can be found in Appendix \ref{sec:res-piv}, and now we describe the main steps in that resolution process.

\subsection{Resolution of singularities}

System (\ref{eq:sistemz}) has no singularities in the affine part of $\mathbb{CP}^2$.
However, at the line $\mathcal{L}_0$ at the infinity, as calculated in Appendix \ref{sec:affine}, the system has three base points: $b_0$, $b_1$, $b_2$, whose coordinates do not depend on $z$.

In the next step, we construct blow ups at points $b_0$, $b_1$, $b_2$.
In the resulting space, we obtain three exceptional lines which we denote by $\mathcal{L}_1$, $\mathcal{L}_2$, $\mathcal{L}_3$ respectively.
The induced flow will have one base point on each of these lines, denote them by $b_3$, $b_4$, $b_5$ respectively. These points are also base points for the autonomous system, as their
 coordinates do not depend on $z$.
See Appendix \ref{sec:b012} for details.

Next, blow ups at points $b_3$, $b_4$, $b_5$ are constructed.
The corresponding exceptional lines are $\mathcal{L}_4$, $\mathcal{L}_5$, $\mathcal{L}_6$.
On each of these three lines, there is a base point of the flow.
We denote them by $b_6$, $b_7$, $b_8$.
The coordinates of these points depend on $z$ and they approach the base points of the autonomous flow as $z\to\infty$.
See Appendix \ref{sec:b345} for details.

Finally, blow ups at $b_6$, $b_7$, $b_8$ show that there are no new base points.
The exceptional lines are denoted by $\mathcal{L}_7(z)$, $\mathcal{L}_8(z)$, $\mathcal{L}_9(z)$.

By this procedure, we constructed the fibers $\mathcal{F}(z)$, $z\in\mathbb{C}\cup\{\infty\}\setminus\{0\}$ of the Okamoto space $\mathcal{O}$ for the system (\ref{eq:sistemz}), see Figure \ref{fig:fiber}.
We denote by $\mathcal{L}_i^*$ the proper preimages of the lines $\mathcal{L}_i$, $0\le i\le6$. 

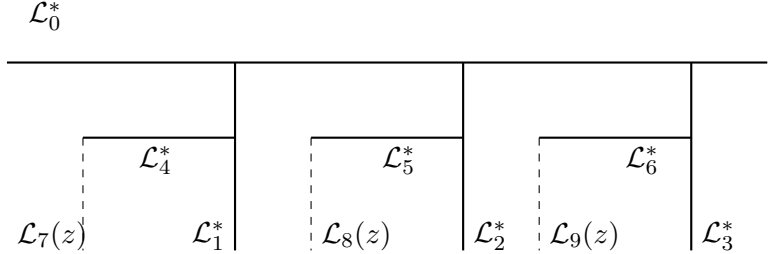
\begin{figure}[ht!] 
\centering
  \begin{tikzpicture}
  \draw[thick] (0,-0.5) -- (10,-0.5);
  		\draw (0.5, -0.2) node[above] {$\mathcal{L}_0^*$};
  \draw[thick]  (3,-0.5) -- (3,-3);
  		\draw (3,-2.8) node[left] {$\mathcal{L}_1^*$};
  \draw[thick]  (6,-0.5) -- (6,-3);
  		\draw (6,-2.8)  node[right] {$\mathcal{L}_2^*$};
  \draw[thick]  (9,-0.5) -- (9,-3);
  		\draw (9,-2.8)  node[right] {$\mathcal{L}_3^*$};
  \draw[thick]  (3,-1.5) -- (1,-1.5);
  		\draw (1.6,-1.8) node[right] {$\mathcal{L}_4^*$};
  \draw[thick]  (6,-1.5) -- (4,-1.5);
  		\draw (4.8,-1.8) node[right] {$\mathcal{L}_5^*$};
  \draw[thick]  (9,-1.5) -- (7,-1.5);
  		\draw (8,-1.8) node[right] {$\mathcal{L}_6^*$};
  \draw[dashed]  (1,-1.5) -- (1,-3);
  		\draw (0,-2.8) node[right] {$\mathcal{L}_7(z)$};
  \draw[dashed]  (4,-1.5) -- (4,-3);
  		\draw (4,-2.8) node[right] {$\mathcal{L}_8(z)$};
  \draw[dashed]  (7,-1.5) -- (7,-3);
  		\draw (7,-2.8) node[right] {$\mathcal{L}_9(z)$};
\end{tikzpicture}
\caption{Fiber $\mathcal{F}(z)$ of the Okamoto space. At the points of $\mathcal{L}_7(z)$, $u$ has a pole and $v$ a zero; on $\mathcal{L}_8(z)$ both have poles; and on $\mathcal{L}_9(z)$, $u$ has a zero and $v$ a pole.}\label{fig:fiber}
\end{figure}


The set where the vector field associated to (\ref{eq:sistem}) becomes infinite will be denoted $\mathcal{I}=\bigcup_{j=0}^6\mathcal{L}_j^*$. 

\subsection{The autonomous system}

The fiber $\mathcal{F}(\infty)$ of the Okamoto space corresponds to the autonomous system obtained by omitting the $z$-dependent terms in (\ref{eq:sistemz}):
\begin{equation}\label{eq:sistemzinf}
\begin{aligned}
u'&=-u(u+2v+2),
\\
v'&=v(2u+v+2),
\end{aligned}
\end{equation}
which is equivalent to
$$
u''=\frac{u'^2}{2u}+\frac{3}{2}u^3+4u^2+2u,
$$
and further to the following family:
$$
(u')^2=\frac12 u^4+2u^3+2u^2+cu, \quad c\in\mathbb{C}.
$$
The solutions of (\ref{eq:sistemzinf}) are thus elliptic funtions.

System (\ref{eq:sistemzinf}) is Hamiltonian, i.e., 
\begin{equation*}
u'=\frac{\partial E}{\partial v},\quad v'=-\frac{\partial E}{\partial u}.
\end{equation*}
where $E=-uv(u+v+2)$.

Note that $b_0$, \dots, $b_5$ are base points of (\ref{eq:sistemzinf}) as well, while $b_6$, $b_7$, $b_8$ will tend to the base points of the autonomous system as $z\to\infty$.

\section{The special solutions}\label{sec:special}
In this section, we analyse singular cubic curves in the pencil parametrised by solutions of the autonomous system (\ref{eq:sistemzinf}) and show that the corresponding solutions of $\PIV$ are either rational or given by parabolic cylinder and exponential functions. 

\subsection{Special solutions and singular cubic curves}\label{sec:special-singular}

The pencil of elliptic curves arising from the Hamiltonian of the autonomous system (\ref{eq:sistemzinf}) is given by the zero set of $h(u,v)=c+uv(u+v+2)$.
For general values of constant $c$, the corresponding curves will be smooth. To investigate singularities, consider the conditions
$$
\frac{\partial h}{\partial u}=0,
\quad
\frac{\partial h}{\partial v}=0,
$$
which give
$$
v(2u+v+2)=0,
\quad
u(u+2v+2)=0.
$$
The solutions are $(0,0)$, $(0,-2)$, $(-2,0)$, which lie on the curve corresponding to $c=0$, and $(-\frac23,-\frac23)$ on the curve corresponding to $c=-\frac{8}{27}$. In other words, there are two singular curves in the pencil and the first (given by $c=0$) contains three singular points, while the second (given by $c=-\frac{8}{27}$) contains one singularity.

Consider first the case $c=0$.
The corresponding curve is $uv(u+v+2)=0$, which is a singular cubic consisting of three lines: $u=0$, $v=0$ and $u+v+2=0$.

\begin{proposition}\label{prop:special}
For $(u,v)$ being a solution of the non-autonomous system (\ref{eq:sistemz}), each derivative of $E=-uv(u+v+2)$ with respect to $z$ vanishes if any of the following three sets of conditions is satisfied:
\begin{enumerate}
 \item $u=0$ and $\alpha_1=0$;
 \item $v=0$ and $\alpha_2=0$;
 \item $u+v+2=0$ and $\alpha_1+\alpha_2=1$.
\end{enumerate}
\end{proposition}
\begin{proof}
We have:
$$
\begin{aligned}
E'&=\frac{dE}{dz}=\frac{\partial E}{\partial u}u'+\frac{\partial E}{\partial v}v'
\\
&=\frac{1}{z}\left(\alpha_2 u( u + 2 v+2) + \alpha_1v( 2 u  + v+2) -uv-\,\frac{3E}{2}\right).
\end{aligned}
$$
{\em Case 1.} Note that $u$ is a divisor of $E$, which is polynomial in $u$ and $v$, and that, when $\alpha_1=0$,  $u$ is also a divisor of $u'$. By induction, it follows that all derivatives of $E$ will be multiples of $u$ and polynomials of $u$, $v$ and thus equal to zero for $u=0$.\hfill\\
{\em Case 2.} The proof is analogous to that of Case 1.\hfill\\
{\em Case 3.} Note that that $E$ is a product of $u+v+2$ and a polynomial of $u$, $v$.
For $\alpha_1+\alpha_2=1$, the derivative of $u+v+2$ is of the same form:
$$
(u+v+2)'=-\frac{(u + v+2) (1 + 2 u z - 2 v z)}{2 z}.
$$
By induction, the same result as in Cases 1 and 2 will hold for all derivatives of $E$. 
\end{proof}

\begin{remark}
$\alpha_1=0$ is equivalent to $\beta=0$; $\alpha_2=0$ to $\beta=-2(1-\alpha)^2$; and $\alpha_1+\alpha_2=1$ to $\beta=-2(1+\alpha)^2$.
\end{remark}

For $\beta=-2(1+\epsilon\alpha)^2$, $\epsilon\in\{-1,1\}$, the Painlev\'e equation (\ref{eq:PIV}) is equivalent to the following Riccati equation:
$$
\frac{dy}{dx}=\epsilon(y^2+2xy)-2(1+\epsilon\alpha),
$$
which can be solved in terms of parabolic cylinder and exponential functions:
\begin{gather*}
y=-\epsilon\frac{d\phi/dx}{\phi},
\\
\phi(x)=\left(C_1U\left(\alpha+\frac{\epsilon}{2},\sqrt{2}x\right)+C_2V\left(\alpha+\frac{\epsilon}{2},\sqrt{2}x\right)\right)e^{\epsilon x^2/2}.
\end{gather*}
Note that a zero of $\phi(x)$ corresponds to a pole of $u(z)$.

For $\epsilon=1$, which is Case 3 of Proposition \ref{prop:special}, $v$ also has a pole, thus each point of $\mathcal{L}_8^*$ corresponds to a special solution.
For $\epsilon=-1$, which is Case 2 of Proposition \ref{prop:special}, $v$ has a zero, thus each point of $\mathcal{L}_7^*$ corresponds to a special solution.
For $\beta=0$, which is Case 1 of Proposition \ref{prop:special}, the solution can be expressed in terms of Hermite polynomials.
Each point of $\mathcal{L}_9^*$ corresponds to such a solution.

Now, consider the case $c=-\frac{8}{27}$.
The corresponding curve is $uv(u+v+2)=\frac{8}{27}$ and it has a unique singular point $(-\frac23,-\frac23)$.
For $\tilde{u}=u+\frac23$ and $\tilde{v}=v+\frac23$, the equation of the curve becomes:
$$
-\frac23(\tilde{u}^2+\tilde{u}\tilde{v}+\tilde{v}^2)+\tilde{u}\tilde{v}(\tilde{u}+\tilde{v})=0,
$$
thus the curve has an ordinary self-intersection at the singular point.
The corresponding solutions are rational.

\subsection{Special rational solutions of $\PIV$}\label{sec:rational}

Consider the following rational solutions of $\PIV$:
\begin{align*}
&y=\pm\frac{1}{x},\quad\text{for}\ \alpha=\pm2,\ \beta=-2;\\
&y=-2x,\quad\text{for}\ \alpha=0,\ \beta=-2;\\
&y=-\frac{2}{3}x,\quad\text{for}\ \alpha=0,\ \beta=-\frac{2}{9}.
\end{align*}
The corresponding solutions of the system (\ref{eq:sistemz}) are:
\begin{align*}
u&=\frac{1}{2z}, & v&=0, & (\alpha_1,\alpha_2)&=(-1,0);\\
u&=-\frac{1}{2z}, & v&=\frac{1}{2z},  & (\alpha_1,\alpha_2)&=(1,1);\\
u&=\frac{1}{2z}, & v&=-2, & (\alpha_1,\alpha_2)&=(1,-1);\\
u&=-\frac{1}{2z}, & v&=\frac{1}{2z}-2, & (\alpha_1,\alpha_2)&=(-1,2);\\
u&=-2, & v&=0,  & (\alpha_1,\alpha_2)&=(1,0);\\
u&=-2, & v&=-\frac{1}{2z},  & (\alpha_1,\alpha_2)&=(-1,1);\\
u&=-\frac{2}{3}, & v&=-\frac23,  & (\alpha_1,\alpha_2)&=\left(\frac13,\frac13\right);\\
u&=-\frac{2}{3}, & v&=-\frac{1}{2z}-\frac23,  & (\alpha_1,\alpha_2)&=\left(-\frac13,\frac23\right).
\end{align*}

All other rational solutions can be obtained from these solutions by B\"acklund transformations \cite{NoumiBOOK}:
\begin{align*}
s_1\ :\ &(u,v;\alpha_1,\alpha_2)\to(u,v+\frac{\alpha_1}{zu};-\alpha_1,\alpha_1+\alpha_2),
\\
s_2\ :\ &(u,v;\alpha_1,\alpha_2)\to(u-\frac{\alpha_2}{zv},v;\alpha_1+\alpha_2,-\alpha_2),
\\
s_3\ :\ &(u,v;\alpha_1,\alpha_2)\to(u-\frac{\alpha_3}{z(u+v+2)},v+\frac{\alpha_3}{z(u+v+2)};1-\alpha_2,1-\alpha_1),
\\
\pi\ :\  &(u,v;\alpha_1,\alpha_2)\to(v,-2-u-v;\alpha_2,1-\alpha_1-\alpha_2),
\end{align*}
which have the following properties:
\begin{gather*}
s_1^2=s_2^2=s_3^2=1,
\quad
(s_1s_2)^3=(s_2s_3)^3=(s_3s_1)^3=1,
\quad
\pi^3=1,
\\
s_2=\pi s_1 \pi^2,
\quad
s_3=\pi s_0 \pi^2,
\quad
s_0=\pi s_3 \pi^2.
\end{gather*}

Also, all special solutions of the fourth Painlev\'e equation can be obtained by the B\"acklund transformations from the solutions mentioned in Section \ref{sec:special-singular}.

\section{The solutions near the infinity set}\label{sec:infinity}
In this section, we will study the behaviour of the solutions of the system (\ref{eq:sistem}) near the set $\mathcal{I}$, where the vector field associated to the system is infinite.

In Lemmas \ref{lemma:E'/E}-\ref{lemma:L3-L0} and Theorem \ref{th:asymptotics}, we prove that $\mathcal{I}$ is repelling, i.e.~the solutions do not intersect it; and, moreover, each solution approaching  sufficiently close to $\mathcal{I}$ at point $z$ will have a pole in a neighbourhood of $z$.

\begin{lemma}\label{lemma:E'/E}
For every $\varepsilon>0$ there exists a neighbourhood $U$ of $\mathcal{L}_0^*$ such that
$$
\left|\frac{{E'}}{E}+\frac{3}{2z}\right|<\varepsilon\quad\text{in}\ U.
$$
For each compact subset $K$ of 
$(\mathcal{L}_1^*\setminus\mathcal{L}_4^*)\cup(\mathcal{L}_2^*\setminus\mathcal{L}_5^*)\cup(\mathcal{L}_3^*\setminus\mathcal{L}_6^*)$,
there exists a neighbourhood $V$ of $K$ and a constant $C>0$ such that:
$$
\left|z\frac{{E'}}{E}\right|<C\quad\text{in}\ V\ \text{for all}\ z\neq0.
$$
\end{lemma}
\begin{proof}
In the charts $(u_{02},v_{02})$ and $(u_{03},v_{03})$ (see Appendix \ref{sec:affine}), the function 
$$r=\frac{{E'}}{E}+\frac{3}{2z}$$ 
is equal to:
\begin{equation*}
\begin{aligned}
r_{02}&=-\frac{u_{02} (\alpha_2 + 2 \alpha_2 u_{02} +(2 \alpha_1  + 2 \alpha_2-1) v_{02} + 2 \alpha_1 u_{02} v_{02} + \alpha_1 v_{02}^2)}
{v_{02} (1 + 2 u_{02} + v_{02}) z},
\\
r_{03}&=-\frac{u_{03}(\alpha_1 + 2 \alpha_1 u_{03} + (2 \alpha_1+ 2 \alpha_2-1) v_{03} + 2 \alpha_2 u_{03} v_{03} + \alpha_2 v_{03}^2)}{v_{03} (1 + 2 u_{03} + v_{03}) z}.
\end{aligned}
\end{equation*}
The first statement of the lemma follows immediately from these expressions, since $\mathcal{L}_0^*$ is given by $u_{02}=0$ and $u_{03}=0$ in those charts, see Section \ref{sec:affine}.

Near $\mathcal{L}_1^*$, in the respective coordinate charts (see Section \ref{sec:b012}), we have
\begin{equation*}
z\frac{E'}{E}+\frac{3}{2} \sim \begin{cases}
				 &-\,\alpha_2\,u_{11}\\
				 &-\,\displaystyle\frac{\alpha_2}{v_{12}}
				 \end{cases}
				 \end{equation*}
Since $\mathcal{L}_4^*$ is given by $v_{12}=0$, see Section \ref{sec:b345}, the statement of the lemma is true for the compact sets $K$ contained in a neighbourhood of 
$\mathcal{L}_1^*\setminus\mathcal{L}_4^*$.

On $\mathcal{L}_2^*$, (see Section \ref{sec:b012}), we have
\begin{equation*}
z\frac{E'}{E} \sim \begin{cases}
		\displaystyle-\frac{3 + 2 (2 + \alpha_1 + \alpha_2) u_{21}}{2 (1 + 2 u_{21})}\\
		\displaystyle-\frac{4 + 2 \alpha_1 + 2 \alpha_2 + 3 v_{22}}{2 (2 + v_{22})}
				 \end{cases}
				 \end{equation*}
Therefore, since $\mathcal{L}_5^*$ is given by the equations $u_{21}=-\frac{1}{2}$ and $v_{22}=-2$, the statement is true for the compacts contained in a neighbourhood of $\mathcal{L}_2^*\setminus\mathcal{L}_5^*$.

On $\mathcal{L}_3^*$,  (see Section \ref{sec:b012}), we have
\begin{equation*}
z\frac{E'}{E} +\frac{3}{2}\sim \begin{cases}
		\displaystyle-\,\alpha_1u_{31}\\
		\displaystyle -\,\frac{\alpha_1}{v_{32}}
				 \end{cases}
				 \end{equation*}
Since $\mathcal{L}_6$ is given by $v_{32}=0$, the statement of the lemma is true for the compact sets $K$ contained in a neighbourhood of 
$\mathcal{L}_3^*\setminus\mathcal{L}_6^*$.  
\end{proof}

\begin{lemma}
There exists a continuous complex valued function $d$ on a neighbourhood of the infinity set $\mathcal{I}$ in the Okamoto space,
such that:
$$
d=
\begin{cases}
\frac{1}{E}, &\text{in a neighbourhood of}\ \mathcal{I}\setminus(\mathcal{L}_4^*\cup\mathcal{L}_5^*\cup\mathcal{L}_6^*),
\\
J_{71}, &\text{in a neighbourhood of}\ \mathcal{L}_4^*\setminus\mathcal{L}_1^*,
\\
\frac{J_{82}}{2}, &\text{in a neighbourhood of}\ \mathcal{L}_5^*\setminus\mathcal{L}_2^*,
\\
-J_{91},&\text{in a neighbourhood of}\ \mathcal{L}_6^*\setminus\mathcal{L}_3^*.
\end{cases}
$$
\end{lemma}

\begin{proof}
Assume $d$ is defined by $\frac{1}{E}$, in a neighbourhood of $\mathcal{I}\setminus(\mathcal{L}_4^*\cup\mathcal{L}_5^*\cup\mathcal{L}_6^*)$.
From Section \ref{sec:b678}, we have that the line $\mathcal{L}_4^*$ is determined by $u_{71}=0$ in the $(u_{71},v_{71})$ chart.
Thus as we approach $\mathcal{L}_4^*$, i.e., as $u_{71}\to0$, we have
$$
EJ_{71}\sim 1 + \frac{\alpha_2}{v_{71} z}
$$
which provides the second result.

From Section \ref{sec:b678}, we have that the line $\mathcal{L}_5^*$ is given by $v_{82}=0$ in the $(u_{82},v_{82})$ chart.
Thus as we approach $\mathcal{L}_5^*$:
$$
EJ_{82} \sim 2 - \frac{1-\alpha_1-\alpha_2}{4u_{82} z}
$$
which gives the third result.

From Section \ref{sec:b678}, we have that the line $\mathcal{L}_6^*$ is given by $u_{91}=0$ in the $(u_{91},v_{91})$ chart.
Then as we approach  $\mathcal{L}_6$:
$$
EJ_{91}\sim -1 + \frac{\alpha_1}{v_{91} z}
$$
which provides the fourth result.  
\end{proof}

\begin{lemma}[Behaviour near $\mathcal{L}_4^*\setminus\mathcal{L}_1^*$]\label{lemma:L4-L1}
If a solution at the complex time $z$ is sufficiently close to $\mathcal{L}_4^*\setminus\mathcal{L}_1^*$, then there exists a unique $\zeta\in\mathbb{C}$ such that:
\begin{enumerate}
\item
$v_{71}(\zeta)=0$, i.e.~ $\zeta\in\mathcal{L}_7$;
\item
$|z-\zeta|=O(|d(z)||v_{71}(z)|)$ for small $d(z)$ and bounded $|v_{71}(z)|$.
\end{enumerate} 
In other words, the solution has a pole at $z=\zeta$.

For large $R_4>0$, consider the set $\{z\in\mathbb{C}\mid |v_{71}|\le R_4\}$.
Then, its connected component containing $\zeta$ is an approximate disk $D_{4}$ with centre $\zeta$ and radius 
$|d(\zeta)|R_4$,
and $z\mapsto v_{71}(z)$ is a complex analytic diffeomorphism from that approximate disk onto $\{v\in\mathbb{C}\mid |v|\le R_4\}$.
\end{lemma}

\begin{proof}
For the study of the solutions near $\mathcal{L}_4^*\setminus\mathcal{L}_1^*$, we use the coordinates $(u_{71},v_{71})$. In this chart, the line $\mathcal{L}_4^*\setminus\mathcal{L}_1^*$ is given by the equation $u_{71}=0$ and parametrized by $v_{71}\in\mathbb{C}$ (see Section \ref{sec:b678}). Moreover, $\mathcal{L}_7^*$ (without one point) is given by $v_{71}=0$ and parametrized by $u_{71}\in\mathbb{C}$. (Equivalent arguments in the alternative chart $(u_{72},v_{72})$ cover the missing point of $\mathcal{L}_7^*$.) 

Asymptotically, for $u_{71}\to0$, and bounded $v_{71}$, ${1}/{z}$, we have
\begin{subequations}
\begin{align}
 v_{71}'&\sim\frac{1}{u_{71}},\label{eq:71-a}
 \\
J_{71}&=-u_{71},\label{eq:71-b}
\\
\frac{J_{71}'}{J_{71}}&=2+\frac{3}{2z}+O(u_{71})=2+\frac{3}{2z}+O(J_{71}),\label{eq:71-c}
\\
EJ_{71}&\sim1+\frac{\alpha_2}{z\,v_{71}}.\label{eq:71-d}
\end{align}
\end{subequations}
Integrating (\ref{eq:71-c}) from $\zeta$ to $z$, we get
$$
J_{71}(z)=J_{71}(\zeta)e^{2(z-\zeta)}\left(\frac{z}{\zeta}\right)^{3/2}(1+o(1)),
\quad
\text{for}\ \frac{z}{\zeta}\sim1.
$$
Hence, using Equation (\ref{eq:71-b}), $u_{71}$ is approximately equal to a small constant, and from (\ref{eq:71-a}) it follows that:
$$
v_{71}(z)\sim v_{71}(\zeta)+\frac{z-\zeta}{u_{71}}.
$$
Thus, if $z$ runs over an approximate disk $D$ centred at $\zeta$ with radius $|u_{71}|R$, then $v_{71}$ fills an approximate disk centred at $v_{71}(\zeta)$ with radius $R$.
Therefore, if $u_{71}(\zeta)\ll 1/\zeta$, for $z\in D$, the solution satisfies
$$
\frac{u_{71}(z)}{u_{71}(\zeta)}\sim1,
$$
and $v_{71}(z)$ is a complex analytic diffeomorphism from $D$ onto an approximate disk with centre $v_{71}(\zeta)$ and radius $R$.
If $R$ is sufficiently large, we will have $0\in v_{71}(D)$, i.e.~ the solution of the Painlev\'e equation will have a pole at a unique point in $D$.

Now, it is possible to take $\zeta$ to be the pole point.
For $|z-\zeta|\ll|\zeta|$, we have:
\begin{gather*}
\frac{d(z)}{d(\zeta)}\sim1,
\quad\text{i.e.}\quad
-\frac{u_{71}(z)}{d(\zeta)}\sim\frac{J_{71}(z)}{d(\zeta)}\sim1,
\end{gather*}
and 
\[
v_{71}(z)\sim \frac{z-\zeta}{u_{71}}\sim-\frac{z-\zeta}{d(\zeta)}.
\]
Let $R_4$ be a large positive real number.
Then the equation $|v_{71}(z)|=R_4$ corresponds to $|z-\zeta|\sim|d(\zeta)|R_4$, which is still small compared to $|\zeta|$ if 
$|d(\zeta)|$ is sufficiently small.
Denote by $D_4$ the connected component of the set of all $z\in\mathbb{C}$ such that $\{z\mid|v_{71}(z)|\le R_4\}$ is an approximate disk with centre $\zeta$ and radius $|d(\zeta)|R_4$.

More precisely, $v_{71}$ is a complex analytic diffeomorphism from $D_4$ onto $\{v\in\mathbb{C}\mid|v|\le R_4\}$, and 
$$
\frac{d(z)}{d(\zeta)}\sim1
\quad\text{for all}\quad 
z\in D_4.
$$
The function $E(z)$ has a simple pole at $z=\zeta$. 
From (\ref{eq:71-d}), we have:
\begin{gather*}
E(z)J_{71}(z)\sim1
\quad
\text{when}
\quad
1\gg\frac{1}{|zv_{71}(z)|}\sim\left|\frac{u_{71}(\zeta)}{\zeta(z-\zeta)}\right|=\frac{|d(\zeta)|}{|\zeta(z-\zeta)|},
\end{gather*}
that is, when
\[
|z-\zeta|\gg\frac{|d(\zeta)|}{|\zeta|}.
\]

We assume $R_4\gg \frac{1}{|\zeta|}$ and, therefore, we have
$$
|z-\zeta|\sim |d(\zeta)|R_4\gg\frac{|d(\zeta)|}{|\zeta|}.
$$
Thus $E(z)J_{71}(z)\sim1$ for the annular disk $z\in D_4\setminus D_4'$, where $D_4'$ is a disk centred at $\zeta$ with small radius compared to the radius of $D_4$.  
\end{proof}

\begin{lemma}[Behaviour near $\mathcal{L}_1^*\setminus\mathcal{L}_0^*$]\label{lemma:L1-L0}
For large finite $R_1>0$, consider the set of all $z\in\mathbb{C}$, such that the solution at complex time $z$ is close to $\mathcal{L}_1^*\setminus\mathcal{L}_0^*$,
with $|v_{41}(z)|\le R_1$, but not close to $\mathcal{L}_4^*$.
Then this set is the complement of $D_4$ in an approximate disk $D_1$ with centre at $\zeta$ and radius $\sim\sqrt{|d(\zeta)|R_1}$.
More precisely, $z\mapsto v_{41}$ defines a $2$-fold covering from the annular domain $D_1\setminus D_4$ onto the complement in
$\{u\in\mathbb{C}\mid|u|\le R_1\}$ of an approximate disk with centre at the origin and small radius $\sim|d(\zeta)|R_4^2$,
where $v_{41}(z)\sim-d(\zeta)(z-\zeta)^2$.
\end{lemma}

\begin{proof}
Set $\mathcal{L}_1^*\setminus\mathcal{L}_0^*$ is visible in the chart $(u_{41},v_{41})$, where it is given by the equation $u_{41}=0$ and parametrized by $v_{41}\in\mathbb{C}$, see Section \ref{sec:b345}.
In that chart, the line $\mathcal{L}_4^*$ (without one point) is given by the equation $v_{41}=0$ and parametrized by $u_{41}\in\mathbb{C}$.

For $u_{41}\to0$ and bounded $v_{41}$ and ${1}/{z}$, we have:
\begin{subequations}
\begin{align}
 u_{41}'&\sim-\frac{1}{v_{41}},\label{eq:41-a}
 \\
v_{41}'&\sim\frac{2}{u_{41}},\label{eq:41-b}
 \\
J_{41}&=-u_{41}^2v_{41},\label{eq:41-c}
\\
EJ_{41}&\sim1,\label{eq:41-d}
\\
\frac{E'}{E}&\sim-\frac{3}{2z}-\frac{\alpha_2}{v_{41}z}.\label{eq:41-e}
\end{align}
\end{subequations}
From (\ref{eq:41-e}) and (\ref{eq:41-a}), we get:
$$
\frac{E'}{E}\sim-\frac{3}{2z}+\frac{\alpha_2}{z}u_{41}'.
$$
Integrating from $z_0$ to $z_1$, we obtain:
$$
\log\frac{E(z_1)}{E(z_0)}
\sim
\log\left(\frac{z_1}{z_0}\right)^{-3/2}+
\alpha_2\left(\frac{u_{41}(z_1)}{z_1}-\frac{u_{41}(z_0)}{z_0}+\int_{z_0}^{z_1}\frac{u_{41}(z)}{z^2}dz\right).
$$
Therefore $E(z_1)/E(z_0)\sim1$, if for all $z$ on the segment from $z_0$ to $z_1$ we have $|z-z_0|\ll|z_0|$ and $|u_{41}(z)|\ll|z_0|$.
We choose $z_0$ on the boundary of $D_4$ from the proof of Lemma \ref{lemma:L4-L1}.
Then we have
$$
\frac{d(\zeta)}{d(z_0)}\sim E(z_0)d(\zeta)\sim E(z_0)J_{71}(z_0)\sim1
\quad\text{and}\quad
|v_{71}(z_0)|=R_4,
$$
which implies that 
$$
|u_{41}|=\left|\frac{1}{v_{71}+\frac{\alpha_2}{z}}\right|\sim\frac{1}{R_4}\ll1.
$$
Furthermore, equations (\ref{eq:41-c}) and (\ref{eq:41-d}) imply that:
$$
|v_{41}(z_0)|=\frac{|J_{41}(z_0)|}{|u_{41}(z_0)|^2}\sim|d(\zeta)|R_4^2,
$$
which is small when $|d(\zeta)|$ is sufficiently small.

Since $D_4$ is an approximate disk with centre $\zeta$ and small radius approximately equal to $|d(\zeta)|R_4$,
and $R_4\gg|\zeta|^{-1}$, we have that $|v_{71}(z)|\ge R_4\gg1$. Writing $z=\zeta+r(z_0-\zeta),\ r\ge1$, where $r\ge 1$, we have $u_{41}(z)\ll 1$ and
$$
\frac{|z-z_0|}{|z_0|}=(r-1)\left|1-\frac{\zeta}{z_0}\right|\ll1
\quad\text{if}\quad
r-1\ll\frac{1}{|1-\frac{\zeta}{z_0}|}.
$$
Then equations (\ref{eq:41-c}), (\ref{eq:41-d}) and $E\sim d(\zeta)^{-1}$ yield
$$
u_{41}^{-1}\sim\left(-\frac{v_{41}}{d(\zeta)}\right)^{1/2},
$$
which in combination with (\ref{eq:41-b}) leads to
$$
\frac{d}{dz}\bigl(v_{41}^{1/2}\bigr)\sim(-d(\zeta))^{-1/2}.
$$
Hence, we get
$$
v_{41}^{1/2}\sim v_{41}(z_0)^{1/2}+(-d(\zeta))^{-1/2}(z-z_0),
$$
and therefore
$$
v_{41}(z)\sim-\frac{(z-z_0)^2}{d(\zeta)}
\quad\text{if}\quad
|z-z_0|\gg|d(\zeta)v_{41}(z_0)|^{1/2}.
$$
For large finite $R_1>0$, the equation $|v_{41}|=R_1$ corresponds to
$|z-z_0|\sim\sqrt{|d(\zeta)|R_1}$,
which is still small compared to $|z_0|\sim|\zeta|$,
and therefore
$|z-\zeta|\le|z-z_0|+|z_0-\zeta|\ll|\zeta|.$
This proves the statement of the lemma.  
\end{proof}

\begin{lemma}[Behaviour near $\mathcal{L}_5^*\setminus\mathcal{L}_2^*$]\label{lemma:L5-L2}
If a solution at the complex time $z$ is sufficiently close to $\mathcal{L}_5^*\setminus\mathcal{L}_2^*$, then there exists a unique $\zeta\in\mathbb{C}$ such that:
\begin{enumerate}
\item
$u_{82}(\zeta)=0$, i.e.~ $\zeta\in\mathcal{L}_8$;
\item
$|z-\zeta|=O(|d(z)||u_{82}(z)|)$ for small $d(z)$ and limited $|u_{82}(z)|$.
\end{enumerate} 
In other words, the solution has a pole at $z=\zeta$.

For large $R_5>0$, consider the set $\{z\in\mathbb{C}\mid |u_{82}|\le R_5\}$.
Then, its connected component containing $\zeta$ is an approximate disk $D_{5}$ with centre $\zeta$ and radius 
$|d(\zeta)|R_5$,
and $z\mapsto u_{82}(z)$ is a complex analytic diffeomorphism from that approximate disk onto $\{u\in\mathbb{C}\mid |u|\le R_5\}$.
\end{lemma}
\begin{proof}
For the study of solutions near $\mathcal{L}_5^*\setminus\mathcal{L}_2^*$, we use the coordinates $(u_{82},v_{82})$.
The line $\mathcal{L}_5^*\setminus\mathcal{L}_2^*$ is given by the equation $v_{82}=0$ and parametrized by $u_{82}\in\mathbb{C}$; see Section \ref{sec:b678}.
Moreover, $\mathcal{L}_8(z)$ (without one point), is given by $u_{82}=0$ and parametrized by $v_{82}\in\mathbb{C}$.
Asymptotically, for $v_{82}\to0$, and bounded $u_{82}$, $1/z$, we have:
\begin{subequations}
\begin{align}
u_{82}'&\sim\frac{8}{v_{82}},\label{eq:82-a}
\\
J_{82}&\sim\frac{v_{82}}{8},\label{eq:82-b}
\\
\frac{J_{82}'}{J_{82}}&=-2+\frac{8-5(\alpha_1+\alpha_2)}{2z}+24u_{82}+\frac{3(1-\alpha_1-\alpha_2)}{2zu_{82}^2},\label{eq:82-c}
\\
EJ_{82}&\sim2-\frac{1-\alpha_1-\alpha_2}{4u_{82}z}.\label{eq:82-d}
\end{align}
\end{subequations}
Integrating (\ref{eq:82-c}) from $\zeta$ to $z$, we get:
\begin{gather*}
J_{82}(z)=J_{81}(\zeta)e^{K(z-\zeta)}\left(\frac{z}{\zeta}\right)^{(8-5(\alpha_1+\alpha_2))/2}(1+o(1)),
\\
K=-2+24u_{82}(\tilde{\zeta})+\frac{3(1-\alpha_1-\alpha_2)}{2\tilde{\zeta}u_{82}^2(\tilde{\zeta})},
\end{gather*}
where $\tilde{\zeta}$ is on the integration path.

Because of (\ref{eq:82-b}), $v_{82}$ is approximately equal to a small constant, and from (\ref{eq:82-a}) follows that:
$$
u_{82}\sim u_{82}(\zeta)+8\frac{(z-\zeta)}{v_{82}(\zeta)}.
$$
Thus, if $z$ runs over an approximate disk $D$ centred at $\zeta$ with radius $\frac18|v_{82}|R$, then $u_{82}$ fills an approximate disk centred at $u_{82}(\zeta)$ with radius $R$.
Therefore, if $v_{82}(\zeta)\ll\zeta$, the solution has the following properties for $z\in D$:
$$
\frac{v_{82}(z)}{v_{82}(\zeta)}\sim1,
$$
and $u_{82}$ is a complex analytic diffeomorphism from $D$ onto an approximate disk with centre $u_{82}(\zeta)$ and radius $R$.
If $R$ is sufficiently large, we will have $0\in u_{82}(D)$, i.e.~ the solution of the Painlev\'e equation  will have a pole at a unique point in $D$.

Not, it is possible to take $\zeta$ to be the pole point.
For $|z-\zeta|\ll|\zeta|$, we have:
\begin{gather*}
\frac{d(z)}{d(\zeta)}\sim1,
\quad\text{i.e.}\quad
\frac{v_{82}(z)}{16d(\zeta)}\sim\frac{J_{82}(z)}{2d(\zeta)}\sim1,
\\
u_{82}(z)\sim\frac{8(z-\zeta)}{v_{82}}\sim\frac{z-\zeta}{2d(\zeta)}.
\end{gather*}
Let $R_5$ be a large positive real number.
Then the equation $|u_{82}(z)|=R_5$ corresponds to $|z-\zeta|\sim2|d(\zeta)|R_5$, which is still small compared to $|\zeta|$ if $|d(\zeta)|$ is sufficiently small.
Denote by $D_5$ the connected component of the set of all $z\in\mathbb{C}$ such that $\{z\mid|u_{82}(z)|\le R_5\}$ is an approximate disk with centre $\zeta$ and radius $2|d(\zeta)|R_5$.
More precisely, $u_{82}$ is a complex analytic diffeomorphism from $D_5$ onto $\{u\in\mathbb{C}\mid|u|\le R_5\}$, and
$$
\frac{d(z)}{d(\zeta)}\sim1
\quad\text{for all}\quad
z\in D_5.
$$
The function $E(z)$ has a simple pole at $z=\zeta$.
From (\ref{eq:82-d}), we have:
\begin{gather*}
E(z)J_{82}(z)\sim2
\quad\text{when}\quad
1\gg\frac{1}{|zu_{82}(z)|}\sim\left|\frac{v_{82}(\zeta)}{8\zeta(z-\zeta)}\right|\sim\frac{|d(\zeta)|}{|z-\zeta|},
\end{gather*}
that is when
$|z-\zeta|\gg\frac{|d(\zeta)|}{|\zeta|}$.

Since $R_5\ll {1}/{|\zeta|}$, the approximate radius of $D_5$ is given by
$$
|d(\zeta)|R_5\gg\frac{|d(\zeta)|}{|\zeta|}
$$
Thus $E(z)J_{82}(z)\sim2$ for $z\in D_5\setminus D_5'$, where $D_5'$ is a disk centred at $\zeta$ with small radius compared to the radius of $D_5$.  
\end{proof}

\begin{lemma}[Behaviour near $\mathcal{L}_2^*\setminus\mathcal{L}_0^*$]
For large finite $R_2>0$, consider the set of all $z\in\mathbb{C}$, such that the solution at complex time $z$ is close to $\mathcal{L}_2^*\setminus\mathcal{L}_0^*$,
with $|u_{52}(z)|\le R_2$, but not close to $\mathcal{L}_5^*$.
Then that set is the complement of $D_5$ in an approximate disk $D_2$ with centre at $\zeta$ and radius $\sim\sqrt{|d(\zeta)|R_2}$.
More precisely, $z\mapsto u_{52}$ defines a $2$-fold covering from the annular domain $D_2\setminus D_5$ onto the complement in
$\{u\in\mathbb{C}\mid|u|\le R_2\}$ of an approximate disk with centre at the origin and small radius $\sim|d(\zeta)|R_5^2$,
where $u_{52}(z)\sim-d(\zeta)(z-\zeta)^2$.
\end{lemma}

\begin{proof}
The line $\mathcal{L}_2^*$ is visible in the coordinate system $(u_{52},v_{52})$, where it is given by the equation $v_{52}=0$ and parametrized by $u_{52}\in\mathbb{C}$; see Section \ref{sec:b345}.
In that chart, line $\mathcal{L}_5^*$ without one point is given by the equation $u_{52}=0$ and parametrized by $v_{52}\in\mathbb{C}$, while the line $\mathcal{L}_0^*$ without one point is given by the equation $u_{52}=\frac12$ and also parametrized by $v_{52}\in\mathbb{C}$.
For $v_{52}\to0$ and bounded $u_{52}$ and $1/z$, we have:
$$
\begin{aligned}
u_{52}'&\sim\frac{4}{v_{52}},
\\
v_{52}'&\sim\frac{2(1-8u_{52})}{u_{52}(2u_{52}-1)},
\\
J_{52}&=-\frac{1}{8}u_{52}(2u_{52}-1)^3v_{52}^2,
\\
EJ_{52}&\sim-2,
\\
\frac{E'}{E}&\sim-\frac{2+\alpha_1+\alpha_2}{2z}-\frac{1-\alpha_1-\alpha_1}{4u_{52}z}.
\end{aligned}
$$

We introduce the following coordinate change for convenience in order to make $\mathcal{L}_0^*$ invisible in the chart:
$$
\tilde{u}_{52}=\frac{u_{52}}{u_{52}-\frac12}.
$$
Now, in the $(\tilde{u}_{52},v_{52})$ coordinate system, $\mathcal{L}_2^*\setminus\mathcal{L}_0^*$ is given by the equation $v_{52}=0$ and parametrized by $\tilde{u}_{52}\in\mathbb{C}$, while the line $\mathcal{L}_5^*$ without one point is given by the equation $\tilde{u}_{52}=0$ and parametrized by $v_{52}\in\mathbb{C}$.

For $v_{52}\to0$ and bounded $\tilde{u}_{52}$ and $\frac{1}{z}$, we have:
\begin{subequations}
\begin{align}
 \tilde{u}_{52}'&\sim-\frac{8(\tilde{u}_{52}-1)^2}{v_{52}},\label{eq:52-a}
 \\
v_{52}'&\sim-\frac{4}{\tilde{u}_{52}}+8-12\tilde{u}_{52},\label{eq:52-b}
 \\
J_{52}&=-\frac{1}{16}\frac{\tilde{u}_{52}v_{52}^2}{(\tilde{u}_{52}-1)^4},\label{eq:52-c}
 \\
 EJ_{52}&\sim-2,\label{eq:52-d}
 \\
\frac{E'}{E}&\sim-\frac{3}{2z}+\frac{1-\alpha_1-\alpha_2}{2\tilde{u}_{52}z}.\label{eq:52-e}
\end{align}
\end{subequations}
We also have:
$$
\begin{aligned}
\tilde{J}_{52}&=\frac{\partial\tilde{u}_{52}}{\partial u}\frac{\partial v_{52}}{\partial v}-\frac{\partial\tilde{u}_{52}}{\partial v}\frac{\partial v_{52}}{\partial u}=\frac{\tilde{u}_{52}v_{52}^2}{8(1-\tilde{u}_{52})^2},\\
E\tilde{J}_{52}&=2(1-\tilde{u}_{52})(2-2\tilde{u}_{52}+\tilde{u}_{52}v_{52}).
\end{aligned}
$$

From (\ref{eq:52-e}) and (\ref{eq:52-b}), we get:
$$
\begin{aligned}
\frac{E'}{E}&\sim-\frac{3}{2z}-\frac{1-\alpha_1-\alpha_2}{8z}v_{52}'+\frac{1-\alpha_1-\alpha_2}{z}-3\frac{1-\alpha_1-\alpha_2}{2z}\tilde{u}_{52}.
\end{aligned}
$$
Integrating from $z_0$ to $z_1$, we obtain:
$$
\begin{aligned}
\log\frac{E(z_1)}{E(z_0)}
\sim&
\log\left(\frac{z_1}{z_0}\right)^{-1/2-\alpha_1-\alpha_2}
\\
&
-\frac{1-\alpha_1-\alpha_2}{8}
\left(
\frac{v_{52}(z_1)}{z_1}-\frac{v_{52}(z_0)}{z_0}+\int_{z_0}^{z_1}\frac{v_{52}(z)}{z^2}dz
\right.
\\
&\qquad\qquad\qquad
\left.
+12\int_{z_0}^{z_1}\frac{\tilde{u}_{52}(z)}{z}dz
\right).
\end{aligned}
$$
Therefore $E(z_1)/E(z_0)\sim1$, if for all $z$ on the segment from $z_0$ to $z_1$ we have $|z-z_0|\ll|z_0|$ and $|v_{52}(z)|\ll|z_0|$,
$|\tilde{u}_{52}(z)|\ll|z_0|$.
We choose $z_0$ on the boundary of $D_5$ from the proof of Lemma \ref{lemma:L5-L2}.
Then we have
$$
\frac{d(\zeta)}{d(z_0)}\sim E(z_0)d(\zeta)\sim E(z_0)\frac{J_{82}(z_0)}2\sim1
\quad\text{and}\quad
|u_{82}(z_0)|=R_5,
$$
which implies that 
$$
|v_{52}|=\left|\frac{1}{u_{82}+\frac{1-\alpha_1-\alpha_2}{8z}}\right|\sim\frac{1}{R_5}\ll1.
$$
Furthermore, equations (\ref{eq:52-c}) and (\ref{eq:52-d}) imply that:
$$
\left|\frac{\tilde{u}_{52}(z_0)}{(\tilde{u}_{52}-1)^4}\right|=\frac{16|J_{52}(z_0)|}{|v_{52}(z_0)|^2}\sim8|d(\zeta)|R_5^2,
$$
which is small when $|d(\zeta)|$ is sufficiently small.

Since $D_5$ is an approximate disk with centre $\zeta$ and small radius $\sim|d(\zeta)|R_5$,
and $R_5\gg|\zeta|^{-1}$, we have that $|u_{82}(z)|\ge R_5\gg1$ hence:
$$
|\tilde{u}_{52}(z)|\ll1\quad \text{if} \quad z=\zeta+r(z_0-\zeta),\ r\ge1,
$$
and
$$
\frac{|z-z_0|}{|z_0|}=(r-1)\left|1-\frac{\zeta}{z_0}\right|\ll1
\quad\text{if}\quad
r-1\ll\frac{1}{|1-\frac{\zeta}{z_0}|}.
$$
Then equations (\ref{eq:52-c}), (\ref{eq:52-d}) and $E\sim d(\zeta)^{-1}$ yield
$$
v_{52}^{-1}\sim\left(-\frac{\tilde{u}_{52}}{d(\zeta)}\right)^{1/2},
$$
which in combination with (\ref{eq:41-b}) leads to
$$
\frac{d\tilde{u}_{52}^{1/2}}{dz}\sim(-d(\zeta))^{-1/2},
$$
hence
$$
\tilde{u}_{52}^{1/2}\sim \tilde{u}_{52}(z_0)^{1/2}+(-d(\zeta))^{-1/2}(z-z_0),
$$
and therefore
$$
\tilde{u}_{52}(z)\sim-\frac{(z-z_0)^2}{d(\zeta)}
\quad\text{if}\quad
|z-z_0|\gg|\tilde{u}_{52}(z_0)|^{1/2}.
$$
For large finite $R_2>0$, the equation $|\tilde{u}_{52}|=R_2$ corresponds to
$|z-z_0|\sim\sqrt{|d(\zeta)R_2}$,
which is still small compared to $|z_0|\sim|\zeta|$,
and therefore
$|z-\zeta|\le|z-z_0|+|z_0-\zeta|\ll|\zeta|.$
This proves the statement of the lemma. 
\end{proof}

\begin{lemma}[Behaviour near $\mathcal{L}_6^*\setminus\mathcal{L}_3^*$]\label{lemma:L6-L3}
If a solution at the complex time $z$ is sufficiently close to $\mathcal{L}_6^*\setminus\mathcal{L}_3^*$, then there exists a unique $\zeta\in\mathbb{C}$ such that:
\begin{enumerate}
\item
$v_{91}(\zeta)=0$, i.e.~ $v_{91}(\zeta)\in\mathcal{L}_9(\zeta)$;
\item
$|z-\zeta|=O(|d(z)||v_{91}(z)|)$ for small $d(z)$ and limited $|v_{91}(z)|$.
\end{enumerate}
In other words, the solution has a pole at $z=\zeta$.

For large $R_6>0$, consider the set $\{z\in\mathbb{C}\mid |v_{91}|\le R_6\}$.
Then, its connected component containing $\zeta$ is an approximate disk $D_{6}$ with centre $\zeta$ and radius 
$|d(\zeta)|R_6$,
and $z\mapsto v_{91}(z)$ is a complex analytic diffeomorphism from that approximate disk onto $\{v\in\mathbb{C}\mid |v|\le R_6\}$.
\end{lemma}

\begin{proof}
Line $\mathcal{L}_6^*\setminus\mathcal{L}_3^*$ is given by the equation $u_{91}=0$ and parametrized by $v_{91}\in\mathbb{C}$, see Section \ref{sec:b678}.
Moreover, $\mathcal{L}_9$ (without one point), is given by $v_{91}=0$ and parametrized by $u_{91}\in\mathbb{C}$.
For the study of the solutions near $\mathcal{L}_6^*\setminus\mathcal{L}_3^*$, we use the coordinates $(u_{91},v_{91})$.
Asymptotically, for $u_{91}\to0$, and bounded $v_{91}$, ${1}/{z}$, we have:
$$
\begin{aligned}
 v_{91}'&\sim-\frac{1}{u_{91}},\label{eq:91-a}
 \\
J_{91}&=u_{91},\label{eq:91-b}
\\
\frac{J_{91}'}{J_{91}}&=-2+\frac{3}{2z}+O(u_{91})=-2+\frac{3}{2z}+O(J_{91}),\label{eq:91-c}
\\
EJ_{91}&\sim-1+\frac{\alpha_1}{zv_{91}}.\label{eq:91-d}
\end{aligned}
$$
Notice that these equations are analogous to (\ref{eq:71-a})-(\ref{eq:71-d}), thus the remainder of the proof is similar to that provided for Lemma \ref{lemma:L4-L1}.  
\end{proof}

\begin{lemma}[Behaviour near $\mathcal{L}_3^*\setminus\mathcal{L}_0^*$]\label{lemma:L3-L0}
For large finite $R_3>0$, consider the set of all $z\in\mathbb{C}$, such that the solution at complex time $z$ is close to $\mathcal{L}_3^*\setminus\mathcal{L}_0^*$,
with $|v_{61}(z)|\le R_1$, but not close to $\mathcal{L}_6^*$.
Then the connected component of that set containing $\zeta$ is the complement of $D_6$ in an approximate disk $D_3$ with centre at $\zeta$ and radius $\sim\sqrt{|d(\zeta)|R_3}$.
More precisely, $z\mapsto v_{61}$ defines a $2$-fold covering from the annular domain $D_3\setminus D_6$ onto the complement in
$\{u\in\mathbb{C}\mid|u|\le R_3\}$ of an approximate disk with centre at the origin and small radius $\sim|d(\zeta)|R_6^2$,
where $v_{61}(z)\sim-d(\zeta)(z-\zeta)^2$.
\end{lemma}

\begin{proof}
The line $\mathcal{L}_3^*\setminus\mathcal{L}_0^*$ is visible in the coordinate system $(u_{61},v_{61})$, where it is given by the equation $u_{61}=0$ and parametrized by $v_{61}\in\mathbb{C}$; see Section \ref{sec:b345}.
In that chart, the line $\mathcal{L}_6^*$ (without one point) is given by the equation $v_{61}=0$ and parametrized by $u_{61}\in\mathbb{C}$.

For $u_{61}\to0$ and bounded $v_{61}$ and ${1}/{z}$, we have:
$$
\begin{aligned}
 u_{61}'&\sim\frac{1}{v_{61}},
 \\
v_{61}'&\sim-\frac{2}{u_{61}},
 \\
J_{61}&=u_{61}^2v_{61},
\\
EJ_{61}&\sim-1,
\\
\frac{E'}{E}&\sim-\frac{3}{2z}-\frac{\alpha_1}{v_{61}z}.
\end{aligned}
$$
Notice that these equations are analogous to (\ref{eq:41-a})-(\ref{eq:41-e}). Therefore, the remainder of the proof is similar to that provided for Lemma \ref{lemma:L1-L0}.  
\end{proof}

\begin{theorem}\label{th:asymptotics}
Let $\varepsilon_1$, $\varepsilon_2$, $\varepsilon_3$ be given such that 
$\varepsilon_1>0$, 
$0<\varepsilon_2<\frac32$, 
$0<\varepsilon_3<1$.
Then there exists $\delta>0$ such that if $|z_0|>\varepsilon_1$ and $|d(z_0)|<\delta$, then:
$$
\rho=\sup\{r>|z_0|\ \text{such that}\ |d(z)|<\delta\ \text{whenever}\ |z_0|\le|z|\le r\}
$$
satisfies:
\begin{enumerate}
\item
$\delta\ge|d(z_0)|\left(\dfrac{\rho}{|z_0|}\right)^{3/2-\varepsilon_2}(1-\varepsilon_3)$;
\item
if $|z_0|\le|z|\le\rho$ then
$d(z)=d(z_0)\left(\dfrac{z}{z_0}\right)^{3/2+\varepsilon_2(z)}(1+\varepsilon_3(z))$;
\item
if $|z|\ge\rho$ then $d(z)\ge\delta(1-\varepsilon_3)$.
\end{enumerate}
\end{theorem}

\begin{proof}
Suppose a solution of the system (\ref{eq:sistemz}) is close to $\mathcal{L}_0^*$ at times $z_0$ and $z_1$.
It follows from Lemmas \ref{lemma:L4-L1}--\ref{lemma:L3-L0} that for every solution close to $\mathcal{I}$, the set of complex times $z$ such that the solution is not close to $\mathcal{L}_0^*$ is the union of approximate disks of radius $\sim|d|^{1/2}$.
Hence if the solution is near $\mathcal{I}$ for all complex times $z$ such that $|z_0|\le|z|\le|z_1|$, then there exists a path $\gamma$ from $z_0$ to $z_1$, such that the solution is close to $\mathcal{L}_0$ for all $z\in\gamma$ and $\gamma$ is $C^1$-close to the path:
$t\mapsto z_1^tz_0^{1-t}$, $t\in[0,1]$,

Then Lemma \ref{lemma:E'/E} implies that:
$$
\log\frac{E(z)}{E(z_0)}=-\frac{3}{2}\log\frac{z}{z_0}\int_0^1dt+o(1),
$$
therefore
$$
E(z)=E(z_0)\left(\frac{z}{z_0}\right)^{3/2+o(1)}(1+o(1)),
$$
and
\begin{equation}\label{eq:d(z)=d(z0)}
d(z)=d(z_0)\left(\frac{z}{z_0}\right)^{3/2+o(1)}(1+o(1)).
\end{equation}
From Lemmas \ref{lemma:L4-L1}--\ref{lemma:L3-L0} we then have that, as long as the solution is close to $\mathcal{I}$, as it moves into a neighbourhood of $\mathcal{L}_4^*\setminus\mathcal{L}_1^*$, $\mathcal{L}_5^*\setminus\mathcal{L}_2^*$, $\mathcal{L}_6^*\setminus\mathcal{L}_3^*$, the ratio of $d$ remains close to $1$.

For the first statement of the theorem, we have:
$$
\delta>d(z)\ge d(z_0)\left(\frac{z}{z_0}\right)^{3/2-\varepsilon_2}(1-\varepsilon_3)
$$
and so
$$
\delta\ge\sup_{\{z\mid|d(z)|<\delta\}} d(z_0)\left(\frac{z}{z_0}\right)^{3/2-\varepsilon_2}(1-\varepsilon_3).
$$
The second statement follows from (\ref{eq:d(z)=d(z0)}), while the third follows by the assumption on $z$. 
\end{proof}

\section{The limit set}\label{sec:limitset}
In this section we define and consider properties of the limit sets of solutions. In Theorem \ref{th:K}, we prove that there is a compact set $K\subset\mathcal{F}(\infty)$, such the limit sets of all solutions of  (\ref{eq:sistem}) are contained in $K$ and that the limit set of any solution is non-empty, compact, connected, and invariant under the flow of the autonomous system (\ref{eq:sistemzinf}).
These results lead us to Theorem \ref{th:poleszeroes}, i.e., that each non-rational solution of the fourth Painlev\'e equation has infinitely many zeroes and poles.
  
\begin{theorem}\label{th:K}
There exists a compact subset $K$ of $\mathcal{F}(\infty)\setminus\mathcal{I}(\infty)$, such that the limit set $\Omega_{(u,v)}$ of any solution $(u,v)$ is contained in $K$.
Moreover, $\Omega_{(u,v)}$ is a non-empty, compact and connected set, which is invariant under the flow of the autonomous system (\ref{eq:sistemzinf}).
\end{theorem}

\begin{proof}
For any positive numbers $\delta$, $r$, let $K_{\delta,r}$ denote the set of all $s\in\mathcal{F}(z)$ such that $|z|\ge r$ and $|d(s)|\ge\delta$.
Since $\mathcal{F}(z)$ is a complex analytic family over $\mathbb{P}^1\setminus\{0\}$ of compact surfaces 
$\mathcal{F}(z)$, $K_{\delta,r}$ is also compact.
Furthermore $K_{\delta,r}$ is disjoint from the union of the infinity sets $\mathcal{I}(z)$, $z\in\mathbb{P}^1\setminus\{0\}$,
and therefore $K_{\delta,r}$ is a compact subset of Okamoto's space $\mathcal{O}\setminus\mathcal{F}(\infty)$.
When $r$ grows to the infinity, the sets $K_{\delta,r}$ shrink to the set 
$$
K_{\delta,\infty}=\{s\in\mathcal{F}(\infty) \mid |d(s)|\ge\delta \}\subset\mathcal{F}(\infty)\setminus\mathcal{I}(\infty),
$$
which is compact.

It follows from Theorem \ref{th:asymptotics} that there exists $\delta>0$ such that for every solution $(u,v)$ there exists $r_0>0$ with the following property: 
$$
(u(z),v(z))\in K_{\delta,r_0}\ \text{for every}\ z \ \text{such that}\ |z|\ge r_0.
$$
In the sequel, we take $r\ge r_0$, when it follows that $(u(z),v(z))\in K_{\delta,r}$ whenever $|z|\ge r$.
Let $Z_r=\{z\in\mathbb{C}\mid |z|\ge r\}$ and let $\Omega_{(u,v),r}$ denote the closure of $(u,v)(Z_r)$ in $\mathcal{O}$.
Since $Z_r$ is connected and $(u,v)$ is continuous, $\Omega_{(u,v),r}$ is also connected.
Since $(u,v)(Z_r)$ is contained in the compact subset $K_{\delta,r}$, its closure $\Omega_{(u,v),r}$ is also contained in $K_{\delta,r}$ and therefore $\Omega_{(u,v),r}$ is a non-empty compact and connected subset of $\mathcal{O}\setminus\mathcal{F}(\infty)$.
The intersection of a decreasing sequence of non-empty compact and connected sets is non-empty, compact, and connected:
therefore, as $\Omega_{(u,v),r}$ decrease to $\Omega_{(u,v)}$ when $r$ grows to the infinity, it follows that $\Omega_{(u,v)}$ is a non-empty, compact, and connected set of $\mathcal{O}$.
Since $\Omega_{(u,v),r}\subset K_{\delta,r}$ for all $r\ge r_0$, and the sets $K_{\delta,r}$ shrink to the compact subset $K_{\delta,\infty}$ of $\mathcal{F}(\infty)\setminus\mathcal{I}(\infty)$ as $r$ grows to the infinity, it follows that $\Omega_{(u,v)}\subset K_{\delta,\infty}$.
This proves the first statement of the theorem with $K=K_{\delta,\infty}$.

Since $\Omega_{(u,v)}$ is the intersection of the decreasing family of compact sets $\Omega_{(u,v),r}$, there exists for every neighbourhood $A$ of $\Omega_{(u,v)}$ in $\mathcal{O}$ and $r>0$ such that $\Omega_{(u,v),r}\subset A$, hence $(u(z),v(z))\in A$ for every $z\in\mathbb{C}$ such that $|z|\ge r$.
If $z_j$ is any sequence in $\mathbb{C}\setminus\{0\}$ such that $|z_j|\to\infty$, then the compactness of $K_{\delta,r}$, in combination with $(u,v)Z_r\subset K_{\delta,r}$, implies that there is a subsequence $j=j(k)\to\infty$ as $k\to\infty$ and an $s\in K_{\delta,r}$, such that:
$$
(u(z_{j(k)}),v(z_{j(k)}))\to s\ \text{as}\ k\to\infty.
$$
Then it follows that $s\in\Omega_{(u,v)}$.

Next, we prove that $\Omega_{(u,v)}$ is invariant under the flow $\Phi^t$ of the autonomous Hamiltonian system.
Let $s\in\Omega_{(u,v)}$ and $z_j$ be a sequence in $\mathbb{C}\setminus\{0\}$ such that $z_j\to\infty$ and $(u(z_j),v(z_j))\to s$.
Since the $z$-dependent vector field of the Bоutroux-Painlev\'e system converges in $C^1$ to the vector field of the autonomous Hamiltonian system as $z\to\infty$, it followс from the continuous dependence on initial data and parameters, that the distance between $(u(z_j+t),v(z_j+t)$ and $\Phi^t(u(z_j),v(z_j))$ converges to zero as $j\to\infty$.
Since $\Phi^t(u(z_j),v(z_j))\to\Phi^t(s)$ and $z_j\to\infty$ as $j\to\infty$, it follows that $(u(z_j+t),v(z_j+t))\to\Phi^t(s)$ and $z_j+t\to\infty$ as $j\to\infty$, hence $\Phi^t(s)\in\Omega_{(u,v)}$.  
\end{proof}

\begin{proposition}\label{prop:intersections}
Every non-special solution $(u(z),v(z))$ intersects each of the pole lines $\mathcal{L}_7$, $\mathcal{L}_8$, $\mathcal{L}_9$ infinitely many times. 
\end{proposition}

\begin{proof}
First, suppose that a solution $(u(z),v(z))$ intersects the union $\mathcal{L}_7\cup\mathcal{L}_8\cup\mathcal{L}_9$ only finitely many times.

According to Theorem \ref{th:K}, the limit set $\Omega_{(u,v)}$ is a compact set in $\mathcal{F}(\infty)\setminus\mathcal{I}(\infty)$.
If $\Omega_{(u,v)}$ intersects one the three pole lines $\mathcal{L}_7$, $\mathcal{L}_8$, $\mathcal{L}_9$ at a point $p$, then there exists arbitrarily large $z$ such that $(u(z),v(z))$ is arbitrarily close to $p$, when the transversality of the vector field to the pole line implies that $(u(\zeta),v(\zeta))\in\mathcal{L}_7\cup\mathcal{L}_8\cup\mathcal{L}_9$ for a unique $\zeta$ near $z$.
As this would imply that $(u(z),v(z))$ intersects $\mathcal{L}_7\cup\mathcal{L}_8\cup\mathcal{L}_9$ has infinitely many times, it follows that $\Omega_{(u,v)}$ is a compact subset of 
$\mathcal{F}(\infty)\setminus(\mathcal{I}({\infty})\cup\mathcal{L}_7({\infty})\cup\mathcal{L}_8({\infty})\cup\mathcal{L}_9({\infty}))$.
However, $\mathcal{L}_7({\infty})\cup\mathcal{L}_8({\infty})\cup\mathcal{L}_9({\infty})$ is equal to the set of all points in 
$\mathcal{F}(\infty)\setminus\mathcal{I}(\infty)$ which project to the line $\mathcal{L}_0$, and therefore 
$\mathcal{F}(\infty)\setminus(\mathcal{I}({\infty})\cup\mathcal{L}_7({\infty})\cup\mathcal{L}_8({\infty})\cup\mathcal{L}_9({\infty}))$
is the affine $(u,v)$ coordinate chart, of which $\Omega_{(u,v)}$ is a compact subset, which implies that $u(z)$ and $v(z)$ remain bounded for large $|z|$.
It follows from boundedness of $u$ and $v$ that $u(z)$ and $v(z)$ are equal to holomorphic functions of $1/z$ in a neighbourhood of $z=\infty$, which implies that there are complex numbers $u(\infty)$, $v(\infty)$ which are the limit points of $u(z)$ and $v(z)$ as 
$|z|\to\infty$. 
In other words, $\Omega_{(u,v)}=\{(u(\infty),v(\infty))\}$ is a one point set.
That means that that the solution is analytic at infinity, i.e.,~ it is analytic on the whole $\mathbb{CP}^1$, thus it must be rational.

Since the limit set $\Omega_{(u,v)}$ is invariant under the autonomous flow, it means that it will contain the whole irreducible component of a cubic curve:
$
-uv(u+v+2)=c,
$
for some constant $c$.
As shown in Section \ref{sec:special}, such a curve is reducible for $c=0$, and the special solutions correspond to each of the irreducible components.
In all other cases, all three base points $b_0$, $b_1$, $b_2$ on the line $\mathcal{L}_0$ will be contained in the limit set, which are projections of the pole lines $\mathcal{L}_7(\infty)$, $\mathcal{L}_8(\infty)$, $\mathcal{L}_9(\infty)$ respectively.
Thus, a non-special solution will intersect each of them infinitely many times.     
\end{proof}

\begin{remark}\label{rem:rational}
The limit set $\Omega_{(u,v)}$ is invariant under the autonomous Hamiltonian system.
If it contains only one point, as we obtained in the proof of Theorem \ref{prop:intersections}, that point must be an equilibrium point of the autonomous Hamiltonian system (\ref{eq:sistemzinf}), that is:
$$
(u(\infty),v(\infty))\in\left\{(0,0),(0,-2),(-2,0),\left(-\frac23,-\frac23\right)\right\}.
$$
These are limiting values of the rational solutions, see Section \ref{sec:rational}.
\end{remark}

\begin{theorem}\label{th:poleszeroes}
Every non-special solution of the fourth Painlev\'e equation (\ref{eq:PIV}) has infinitely many poles and infinitely many zeros.
\end{theorem}
\begin{proof}
It is enough to prove that a non-special solution $u$ of (\ref{eq:sistemz}) has infinitely many poles and zeroes.
Notice that at the intersection point with $\mathcal{L}_7$, $u$ has a pole and $v$ a zero; at the intersection with $\mathcal{L}_8$ both have poles, and on $\mathcal{L}_9$, $u$ has a zero and $v$ a pole.
Since it is shown in Propostion \ref{prop:intersections} that $(u,v)$ intersects each of the lines  $\mathcal{L}_7$, $\mathcal{L}_8$, $\mathcal{L}_9$ infinitely many times, the statement is proved.  
\end{proof}

\begin{appendices}


\section{Resolution of the Painlev\'e vector field}\label{sec:res-piv}
\subsection{The affine charts}\label{sec:affine}
\subsubsection{Affine Chart $(u_{01}, v_{01})$}
The first affine chart is defined by the original coordinates
$$
\begin{aligned}
 u_{01}&=u,
 \\
 v_{01}&=v,
 \\
 E&=-uv(u+v+2).
\end{aligned}
$$

\subsubsection{Affine Chart $(u_{02}, v_{02})$} The second affine chart is given by the following coordinates:
\begin{gather*}
u_{02}=\frac1{u},
\qquad
v_{02}=\frac{v}{u},
\\
u=\frac1{u_{02}},
\qquad
v=\frac{v_{02}}{u_{02}}.
\end{gather*}
The line at the infinity is $\mathcal{L}_0: u_{02}=0$.

The Painlev\'e vector field is given by
$$
\begin{aligned}
u_{02}'&=1+2u_{02}+2v_{02}+\frac1{2z}(2\alpha_1u_{02}^2+u_{02}),
\\
v_{02}'&=\frac{v_{02}}{u_{02}}(4u_{02}+3v_{02}+3)+\frac1{z}(-\alpha_2u_{02}+\alpha_1u_{02}v_{02}),
\end{aligned}
$$
which contain base points at
$$
b_0\ :\ u_{02}=0,v_{02}=0
\quad\text{and}\quad
b_1\ :\ u_{02}=0,v_{02}=-1.
$$

The energy is
$$
\begin{aligned}
E&=-\frac{v_{02} (1 + 2 u_{02} + v_{02})}{u_{02}^3},
\\
E'&=\frac{1}{2 u_{02}^3 z}
\left(4 \alpha_1 u_{02}^2 v_{02} + 2 \alpha_1 u_{02} v_{02}^2+ 4 \alpha_2 u_{02}^2+  3 v_{02}^2
\right.
\\
&\qquad\qquad
\left.
+ 4(\alpha_1+\alpha_2+1) u_{02} v_{02}
+2 \alpha_2 u_{02}  + 3 v_{02}\right).
\end{aligned}
$$

\subsubsection{Affine Chart $(u_{03}, v_{03})$}
We have the coordinates
\begin{gather*}
u_{03}=\frac1{v},
\qquad
v_{03}=\frac{u}{v},
\\
u=\frac{v_{03}}{u_{03}},
\qquad
v=\frac{1}{u_{03}},
\end{gather*}
and the line at the infinity is given by $\mathcal{L}_0: u_{03}=0$.

The flow is given by
$$
\begin{aligned}
u_{03}'&=-1-2u_{03}-2v_{03}+\frac1{2z}(2\alpha_2u_{03}^2+u_{03}),
\\
v_{03}'&=-\frac{v_{03}}{u_{03}}(4u_{03}+3v_{03}+3)+\frac1z(-\alpha_1u_{03}+\alpha_2u_{03}v_{03}),
\end{aligned}
$$
which contains a base point at
$$
b_2\ :\ u_{03}=0,v_{03}=0,
$$
and $(u_{03}=0,v_{03}=-1)$, which is $b_1$.

The energy is given by
$$
\begin{aligned}
E&=-\frac{v_{03} (1 + 2 u_{03} + v_{03})}{u_{03}^3},
\\
E'&=\frac{1}{2 u_{03}^3 z}
\left(4 \alpha_2 u_{03}^2 v_{03}+ 2 \alpha_2 u_{03} v_{03}^2+ 4 \alpha_1 u_{03}^2 
 +  3 v_{03}^2
\right. 
\\
&\qquad\qquad
\left.
+4(\alpha_1+\alpha_2+1)u_{03} v_{03}
+2 \alpha_1 u_{03}  + 3 v_{03}
\right).
\end{aligned}
$$

\subsection{Resolution at base points $b_0$, $b_1$, $b_2$ }\label{sec:b012}
\subsubsection{Resolution at $b_0$}
The first chart is given by
the coordinate change:
\begin{gather*}
u_{11}=\frac{u_{02}}{v_{02}}=\frac1{v},
\qquad 
v_{11}=v_{02}=\frac{v}{u},
\\
u=\frac{1}{u_{11}v_{11}},
\qquad
v=\frac1{u_{11}}.
\end{gather*}
The exceptional line is $\mathcal{L}_1:v_{11}=0$.
The preimage of line $\mathcal{L}_0$ is visible in this chart, and given by the equation $u_{11}=0$.

The flow in this chart:
$$
\begin{aligned}
u_{11}'&=-\frac{1}{v_{11}}(v_{11}+2u_{11}v_{11}+2)+\frac{u_{11}}{2z}(1+2\alpha_2u_{11}),
\\
v_{11}'&=\frac{1}{u_{11}}(3v_{11}+4u_{11}v_{11}+3)+\frac{u_{11}v_{11}}{z}(-\alpha_2+\alpha_1v_{11}),
\end{aligned}
$$
contains no new base points.

The energy is given by
$$
\begin{aligned}
E&=-\frac{1 + v_{11} + 2 u_{11} v_{11}}{u_{11}^3 v_{11}^2},
\\
E'&=\frac{1}{2 u_{11}^3 v_{11}^2 z}
\left(3 + 2 \alpha_2 u_{11} + 3 v_{11} + 4(1+\alpha_1+\alpha_2) u_{11} v_{11} 
\right.
\\&
\qquad\qquad
\left.
+ 
 4 \alpha_2 u_{11}^2 v_{11} + 2 \alpha_1 u_{11} v_{11}^2 + 4 \alpha_1 u_{11}^2 v_{11}^2
 \right).
\end{aligned}
$$

The second chart is given by
\begin{gather*}
u_{12}=u_{02}=\frac1{u},
\qquad
v_{12}=\frac{v_{02}}{u_{02}}=v,
\\
u=\frac{1}{u_{12}},
\qquad
v=v_{12}.
\end{gather*}
The exceptional line is $\mathcal{L}_1:u_{12}=0$.
The preimage of line $\mathcal{L}_0$ is not visible in this chart.

The flow is
$$
\begin{aligned}
u_{12}'&=1+2u_{12}+2u_{12}v_{12}+\frac{u_{12}}{2z}(1+2\alpha_1u_{12}),
\\
v_{12}'&=\frac{v_{12}}{u_{12}}(2u_{12}+2+u_{12}v_{12})-\frac{1}{2z}(2\alpha_2+v_{12}).
\end{aligned}
$$
Both the vector field and the anticanonical pencil have base point at
$$
b_3\ :\ u_{12}=0,v_{12}=0.
$$

The energy is given by
$$
\begin{aligned}
E&=-\frac{v_{12} (1 + 2u_{12} + u_{12} v_{12})}{u_{12}^2},
\\
E'&=\frac{1}{2 u_{12}^2 z}
\left(
2 \alpha_2 + 4 \alpha_2 u_{12} + 3 v_{12} + 4(1+ \alpha_1+ \alpha_2) u_{12} v_{12}
\right.
\\
&\qquad\qquad
\left.
+ 
 4 \alpha_1 u_{12}^2 v_{12} + 3 u_{12} v_{12}^2 + 2 \alpha_1 u_{12}^2 v_{12}^2
 \right).
\end{aligned}
$$

\subsubsection{Resolution at $b_1$}
The first chart is given by
the coordinate change:
\begin{gather*}
u_{21}=\frac{u_{02}}{v_{02}+1}=\frac1{u+v},
\qquad
v_{21}=v_{02}+1=\frac{u+v}{u},
\\
u=\frac{1}{u_{21}v_{21}},
\qquad
v=\frac{v_{21}-1}{u_{21}v_{21}}.
\end{gather*}
The exceptional line is $\mathcal{L}_2:v_{21}=0$.
The preimage of the line $\mathcal{L}_0$ is visible in this chart, and given by the equation $u_{21}=0$.

The flow is given by
$$
\begin{aligned}
u_{21}'&=\frac{(2u_{21}+1)(2-v_{21})}{v_{21}}+\frac{u_{21}}{2z}\big(2(\alpha_1+\alpha_2)u_{21}+1\big),
\\
v_{21}'&=\frac{(4u_{21}+3)(v_{21}-1)}{u_{21}}+\frac{u_{21}v_{21}}{z}(\alpha_1v_{21}-\alpha_1-\alpha_2),
\end{aligned}
$$
and contains a new base point at
$$
b_4 :\ u_{21}=-\frac12,v_{21}=0.
$$

The energy is given by
$$
\begin{aligned}
E&=-\frac{( 2 u_{21} +1) (v_{21}-1)}{u_{21}^3 v_{21}^2},
\\
E'&=\frac{1}{2 u_{21}^3 v_{21}^2 z}\big({-3} - 2(2+ \alpha_1 + \alpha_2) u_{21} + 3 v_{21} + 4 (1+ \alpha_2) u_{21} v_{21}
\\
&\qquad\qquad
+ 4( \alpha_2 - \alpha_1) u_{21}^2 v_{21} + 2 \alpha_1 u_{21} v_{21}^2 + 4 \alpha_1 u_{21}^2 v_{21}^2\big).
\end{aligned}
$$

The second chart is given by 
\begin{gather*}
u_{22}=u_{02}=\frac1{u},
\qquad
v_{22}=\frac{v_{02}+1}{u_{02}}=u+v,
\\
u=\frac{1}{u_{22}},
\qquad
v=v_{22}-\frac{1}{u_{22}}.
\end{gather*}
The exceptional line is $\mathcal{L}_2:u_{22}=0$.
The preimage of line $\mathcal{L}_0$ is not visible in this chart.

The flow is given by
$$
\begin{aligned}
u_{22}'&=-1+2u_{22}+2u_{22}v_{22}+\frac{u_{22}}{2z}(1+2\alpha_1u_{22}),
\\
v_{22}'&=\frac{(v_{22}+2)(u_{22}v_{22}-2)}{u_{22}}-\frac{1}{2z}(v_{22}+2\alpha_1+2\alpha_2),
\end{aligned}
$$
and contains a base point $(u_{22}=0,v_{22}=-2)$, which is $b_4$.

The energy is given by
$$
\begin{aligned}
E&=\frac{(2 + v_{22}) (1 - u_{22} v_{22})}{u_{22}^2},
\\
E'&=\frac{1}{2 u_{22}^2 z}\big({-2}(2+ \alpha_1 + \alpha_2) + 4( \alpha_2- \alpha_1) u_{22} - 3 v_{22} 
+ 4(1+\alpha_2) u_{22} v_{22}
\\
&\qquad\qquad
+ 4 \alpha_1 u_{22}^2 v_{22} + 3 u_{22} v_{22}^2 + 2 \alpha_1 u_{22}^2 v_{22}^2\big).
\end{aligned}
$$

\subsubsection{Resolution at $b_2$}
The first chart is given by
\begin{gather*}
u_{31}=\frac{u_{03}}{v_{03}}=\frac1u,
\qquad
v_{31}=v_{03}=\frac{u}{v},
\\
u=\frac{1}{u_{31}},
\qquad
v=\frac{1}{u_{31}v_{31}}.
\end{gather*}
The exceptional line is $\mathcal{L}_3:v_{31}=0$.
The preimage of line $\mathcal{L}_0$ is visible in this chart, and given by the equation $u_{31}=0$.

The flow 
$$
\begin{aligned}
u_{31}'&=\frac{v_{31}+2u_{31}v_{31}+2}{v_{31}}+\frac{u_{31}}{2z}(1+2\alpha_1u_{31}),
\\
v_{31}'&=-\frac{3+4u_{31}v_{31}+3v_{31}}{u_{31}}+\frac{u_{31}v_{31}}{z}(-\alpha_1+\alpha_2v_{31}),
\end{aligned}
$$
contains no base point.

The energy is given by
$$
\begin{aligned}
E&=-\frac{1 + v_{31} + 2 u_{31} v_{31}}{u_{31}^3 v_{31}^2},
\\
E'&=\frac{1}{2 u_{31}^3 v_{31}^2 z}
\left(3 + 2 \alpha_1 u_{31} + 3 v_{31} + 4(1+\alpha_1+\alpha_2) u_{31} v_{31} 
\right.
\\
&\qquad\qquad
\left.
+ 
 4 \alpha_1 u_{31}^2 v_{31} + 2 \alpha_2 u_{31} v_{31}^2 + 4 \alpha_2 u_{31}^2 v_{31}^2
 \right)
.
\end{aligned}
$$

The second chart is given by
\begin{gather*}
u_{32}=u_{03}=\frac1v,
\qquad
v_{32}=\frac{v_{03}}{u_{03}}=u,
\\
u=v_{32},
\qquad
v=\frac{1}{u_{32}}.
\end{gather*}
The exceptional line is $\mathcal{L}_3:u_{32}=0$.
The preimage of line $\mathcal{L}_0$ is not visible in this chart.

The flow is
$$
\begin{aligned}
u_{32}'&=-1-2u_{32}-2u_{32}v_{32}+\frac{u_{32}}{2z}(1+2\alpha_2u_{32}),
\\
v_{32}'&=-\frac{v_{32}}{u_{32}}(2u_{32}+2+v_{32}u_{32})-\frac1{2z}(2\alpha_1+v_{32}).
\end{aligned}
$$
Both the vector field and the anticanonical pencil have a base point at
$$
b_5\ :\ u_{32}=0,v_{32}=0.
$$

The energy is given by
$$
\begin{aligned}
E&=-\frac{v_{32} (1 + 2u_{32} + u_{32}v_{32})}{u_{32}^2},
\\
E'&=\frac{1}{2 u_{32}^2 z}
\left(2 \alpha_1 + 4 \alpha_1 u_{32} + 3 v_{32} 
+ 4(1+ \alpha_1+ \alpha_2) u_{32} v_{32}
\right.
\\
&\qquad\qquad
\left.
+ 
 4 \alpha_2 u_{32}^2 v_{32} + 3 u_{32} v_{32}^2 + 2 \alpha_2 u_{32}^2 v_{32}^2
 \right).
\end{aligned}
$$

\subsection{Resolution at points $b_3$, $b_4$, $b_5$}\label{sec:b345}

\subsubsection{Resolution at $b_3$}\label{sec:b3}
The first chart is 
\begin{gather*}
u_{41}=\frac{u_{12}}{v_{12}}=\frac1{uv},
\qquad
v_{41}=v_{12}=v,
\\
u=\frac{1}{u_{41}v_{41}},
\qquad
v=v_{41},
\end{gather*}
and the corresponding Jacobian is
\[
J_{41}=\frac{\partial u_{41}}{\partial u}\frac{\partial v_{41}}{\partial v}-\frac{\partial u_{41}}{\partial v}\frac{\partial v_{41}}{\partial u}
=-\frac{1}{u^2v}=-u_{41}^2v_{41}.
\]
The exceptional line is $\mathcal{L}_4:v_{41}=0$.
The preimage of line $\mathcal{L}_1$ in this chart is $u_{41}=0$.
$\mathcal{L}_0$ is not visible in this chart.

The flow is given by
$$
\begin{aligned}
u_{41}'=&-\frac{1}{v_{41}}+u_{41}v_{41}+\frac{u_{41}(\alpha_2+v_{41}+\alpha_1u_{41}v_{41}^2)}{zv_{41}},
\\
v_{41}'=&\frac2{u_{41}}+2v_{41}+v_{41}^2-\frac{1}{2z}(2\alpha_2+v_{41}),
\end{aligned}
$$
and contains a base point:
$$
b_6\ :\ u_{41}=\frac{z}{\alpha_2},v_{41}=0.
$$

The energy and related quantities are
$$
\begin{aligned}
E=&-\frac{1 + 2u_{41} v_{41} + u_{41}v_{41}^2}{u_{41}^2 v_{41}},
\qquad
EJ_{41}=1 + 2u_{41} v_{41} + u_{41}v_{41}^2,
\\
E'=&\frac{1}{2 u_{41}^2 v_{41}^2 z}
\left(2 \alpha_2 + 3 v_{41} + 4 \alpha_2 u_{41} v_{41} 
+ 4(1+\alpha_1+ \alpha_2) u_{41} v_{41}^2
\right.
\\
&
\left.
\qquad\qquad
+ 3 u_{41} v_{41}^3 + 4 \alpha_1 u_{41}^2 v_{41}^3 + 2 \alpha_1 u_{41}^2 v_{41}^4
\right).
\end{aligned}
$$

The second chart is given by
\begin{gather*}
u_{42}=u_{12}=\frac1u,
\qquad
v_{42}=\frac{v_{12}}{u_{12}}=uv,
\\
u=\frac{1}{u_{42}},
\qquad
v=u_{42}v_{42}.
\end{gather*}
The exceptional line is $\mathcal{L}_4:u_{42}=0$.
The preimages of $\mathcal{L}_0$  and $\mathcal{L}_1$ are not visible in this chart.

The flow is
$$
\begin{aligned}
u_{42}'&=1+2u_{42}+2u_{42}^2v_{42}+\frac{u_{42}}{2z}(1+2\alpha_1u_{42}),
\\
v_{42}'&=\frac{v_{42}}{u_{42}}-u_{42}v_{42}^2-\frac{1}{zu_{42}}(\alpha_2+u_{42}v_{42}+\alpha_1u_{42}^2v_{42}),
\end{aligned}
$$
and contains a base point $(u_{42}=0,v_{42}=\frac{\alpha_2}z)$, which is $b_6$.

\subsubsection{Resolution at $b_4$}
The first chart is given by
\begin{gather*}
u_{51}=\frac{u_{21}+\frac12}{v_{21}}=\frac{u(u+v+2)}{2(u+v)^2},
\qquad
v_{51}=v_{21}=\frac{u+v}{u},
\\
u=\frac{2}{v_{51} (2 u_{51} v_{51}-1)},
\qquad
v=\frac{2 (v_{51}-1)}{v_{51} ( 2 u_{51} v_{51}-1)}.
\end{gather*}
The exceptional line is $\mathcal{L}_5:v_{51}=0$.
The preimage of $\mathcal{L}_2$ is not visible in this chart, while the preimage of $\mathcal{L}_0$ is given by $u_{51}v_{51}=\frac{1}{2}$.

The flow is given by
$$
\begin{aligned}
u_{51}'=&-\frac{2u_{51}(1+2u_{51}v_{51}(3v_{51}-4))}{v_{51}(2u_{51}v_{51}-1)}+
\\
&+\frac{(2u_{51}v_{51}-1)(\alpha_1+\alpha_2-1-4(\alpha_1+\alpha_2)u_{51}v_{51}+2\alpha_1u_{51}v_{51}^2)}{4zv_{51}},
\\
v_{51}'=&\frac{2(v_{51}-1)(4u_{51}v_{51}+1)}{2u_{51}v_{51}-1}+\frac{v_{51}(\alpha_1v_{51}-\alpha_1-\alpha_2)(2u_{51}v_{51}-1)}{2z},
\end{aligned}
$$
and contains a base point 
$$
b_7\ :\ u_{51}=\frac{1-\alpha_1-\alpha_2}{8z},v_{51}=0.
$$

The second chart is
\begin{gather*}
u_{52}=u_{21}+\frac{1}{2}=\frac{1}{u+v}+\frac12,
\qquad
v_{52}=\frac{v_{21}}{u_{21}+\frac12}=\frac{2(u+v)^2}{u(u+v+2)},
\\
u=\frac{2}{(2 u_{52}-1) u_{52}v_{52}},
\qquad
v=\frac{2 ( u_{52} v_{52}-1)}{( 2 u_{52}-1)u_{52} v_{52}}
\end{gather*}
which has Jacobian
\[
J_{52}=\frac{\partial u_{52}}{\partial u}\frac{\partial v_{52}}{\partial v}-\frac{\partial u_{52}}{\partial v}\frac{\partial v_{52}}{\partial u}
=-\frac{2}{u^2(u+v+2)}=-\frac{1}{8}u_{52}(2u_{52}-1)^3v_{52}^2.
\]
The exceptional line is $\mathcal{L}_5:u_{52}=0$.
In this chart, the preimage of $\mathcal{L}_2$ is given by $v_{52}=0$, and of $\mathcal{L}_0$ by $u_{52}=\frac{1}{2}$.

The flow is given by
$$
\begin{aligned}
u_{52}'=&-2u_{52}+\frac4{v_{52}}+\frac{(2u_{52}-1)(1+(\alpha_1+\alpha_2)(2u_{52}-1))}{4z},
\\
v_{52}'=&\frac{2(1-8u_{52}+6u_{52}^2v_{52})}{u_{52}(2u_{52}-1)}
\\&
+
\frac{v_{52}(2u_{52}-1)(\alpha_1+\alpha_2-1-4(\alpha_1+\alpha_2)u_{52}+2\alpha_1u_{52}^2v_{52})}{4zu_{52}},
\end{aligned}
$$
which contains a base point
$\left(u_{52}=0, v_{52}={8z}/(1-\alpha_1-\alpha_2)\right)$, which is $b_7$.

The energy and related quantities are 
$$
\begin{aligned}
E=&-\frac{16 ( u_{52} v_{52}-1)}{u_{52} (2 u_{52}-1)^3 v_{52}^2},
\qquad
EJ_{52}=2 ( u_{52} v_{52}-1),
\\
E'=&\frac{4}{u_{52}^2 ( 2 u_{52}-1)^3 v_{52}^2 z}\big((\alpha_1 + \alpha_2-1) - 2(2+ \alpha_1+ \alpha_2) u_{52}
\\
&\qquad
+ (1-\alpha_1-\alpha_2) u_{52} v_{52}
+ 4(1+\alpha_1) u_{52}^2 v_{52}
+4(\alpha_2-\alpha_1) u_{52}^3 v_{52}
\\
&\qquad
- 2 \alpha_1 u_{52}^3 v_{52}^2 + 4 \alpha_1 u_{52}^4 v_{52}^2\big).
\end{aligned}
$$

\subsubsection{Resolution at $b_5$}
The first chart is
\begin{gather*}
u_{61}=\frac{u_{32}}{v_{32}}=\frac{1}{uv},
\qquad
v_{61}=v_{32}=u,
\\
u=v_{61},
\qquad
v=\frac{1}{u_{61}v_{61}},
\\
J_{61}=\frac{\partial u_{61}}{\partial u}\frac{\partial v_{61}}{\partial v}-\frac{\partial u_{61}}{\partial v}\frac{\partial v_{61}}{\partial u}
=\frac{1}{uv^2}=u_{61}^2v_{61}.
\end{gather*}
The exceptional line is $\mathcal{L}_6:v_{61}=0$.
In this chart, the preimage of $\mathcal{L}_3$ is given by $u_{61}=0$, and the preimage of $\mathcal{L}_0$ is not visible.

The flow is
$$
\begin{aligned}
u_{61}'&=\frac{1-u_{61}v_{61}^2}{v_{61}}+\frac{u_{61}(v_{61}+\alpha_1+\alpha_2u_{61}v_{61}^2)}{zv_{61}},
\\
v_{61}'&=-\frac{2+u_{61}v_{61}(2+v_{61})}{u_{61}}-\frac{v_{61}+2\alpha_1}{2z},
\end{aligned}
$$
and contains a base point:
$$
b_8\ :\ u_{61}=-\frac{z}{\alpha_1},v_{61}=0.
$$

The energy is given by
$$
\begin{aligned}
E&=-\frac{1 + 2 u_{61} v_{61} + u_{61} v_{61}^2}{u_{61}^2 v_{61}},
\qquad
EJ_{61}=-(1 + 2 u_{61} v_{61} + u_{61} v_{61}^2),
\\
E'&=\frac{1}{2 u_{61}^2 v_{61}^2 z}
\big(2 \alpha_1 + 3 v_{61} + 4 \alpha_1 u_{61} v_{61} 
+ 4(1+\alpha_1+\alpha_2) u_{61} v_{61}^2
\\
&\qquad\qquad
+ 3 u_{61} v_{61}^3 + 4 \alpha_2 u_{61}^2 v_{61}^3 + 2 \alpha_2 u_{61}^2 v_{61}^4
\big).
\end{aligned}
$$

The second chart is
\begin{gather*}
u_{62}=u_{32}=\frac1v,
\qquad
v_{62}=\frac{v_{32}}{u_{32}}=uv,
\\
u=u_{62}v_{62},
\qquad
v=\frac{1}{u_{62}}.
\end{gather*}
The exceptional line is $\mathcal{L}_6:u_{62}=0$.
In this chart, the preimages of $\mathcal{L}_3$ and $\mathcal{L}_0$ are not visible.

The flow is
$$
\begin{aligned}
u_{62}'&=-1-2u_{62}-2u_{62}^2v_{62}+\frac{u_{62}(2\alpha_2u_{62}+1)}{2z},
\\
v_{62}'&=\frac{v_{62}(-1+u_{62}^2v_{62})}{u_{62}}-\frac{\alpha_1+u_{62}v_{62}(1+\alpha_2u_{62})}{zu_{62}},
\end{aligned}
$$
and contains a base point is $u_{62}=0,v_{62}=-\,{\alpha_1}/{z}$, which is $b_8$.

The energy is given by
$$
\begin{aligned}
E&=-\frac{v_{62}}{u_{62}} (1 + 2 u_{62} + u_{62}^2 v_{62}),
\qquad
EJ_{62}=-v_{62}(1 + 2 u_{62} + u_{62}^2 v_{62}),
\\
E'&=\frac{1}{2 u_{62}^2 z}
\big(2 \alpha_1 + 4 \alpha_1 u_{62} + 3 u_{62} v_{62} 
+ 4(1+\alpha_1+\alpha_2) u_{62}^2 v_{62} 
\\
&\qquad\qquad
+ 4 \alpha_2 u_{62}^3 v_{62} + 3 u_{62}^3 v_{62}^2 + 2 \alpha_2 u_{62}^4 v_{62}^2\big).
\end{aligned}
$$

\subsection{Resolution at points $b_6$, $b_7$, $b_8$}\label{sec:b678}

\subsubsection{Resolution at $b_6$}
The first chart is
\begin{gather*}
u_{71}=\frac{u_{42}}{v_{42}-\frac{\alpha_2}{z}}=\frac{z}{u(uvz-\alpha_2)},
\qquad
v_{71}=v_{42}-\frac{\alpha_2}{z}=uv-\frac{\alpha_2}{z},
\\
u =\frac{1}{u_{71} v_{71}},
\qquad
v=u_{71} v_{71} \left(v_{71}+\frac{\alpha_2}{z} \right),
\end{gather*}
which gives the Jacobian 
\begin{align*}
J_{71}&=\frac{\partial u_{71}}{\partial u}\frac{\partial v_{71}}{\partial v}-\frac{\partial u_{71}}{\partial v}\frac{\partial v_{71}}{\partial u}
=\frac{z}{u(\alpha_2-uvz)}=-u_{71},
\\
J_{71}'&=-2u_{71}-3u_{71}^2v_{71}^2
-\frac{u_{71}}{2z}\big(3+4(\alpha_1+2\alpha_2)u_{71}v_{71}\big)
-\frac{(\alpha_1+\alpha_2)\alpha_2}{z^2}u_{71}^2.
\end{align*}
The exceptional line is $\mathcal{L}_7:v_{71}=0$.
In this chart, the preimage of $\mathcal{L}_4$ is given by equation $u_{71}=0$, while the preimages of $\mathcal{L}_1$ and $\mathcal{L}_0$ are not visible.

The flow is given by
$$
\begin{aligned}
u_{71}'=&2u_{71}+3u_{71}^2v_{71}^2
+\frac{u_{71}}{2z}\big(3+4(\alpha_1+2\alpha_2)u_{71}v_{71}\big)
+\frac{(\alpha_1+\alpha_2)\alpha_2}{z^2}u_{71}^2,
\\
v_{71}'=&
\frac{1}{u_{71}} - u_{71} v_{71}^3
-\frac{v_{71}}{z}\big(1 + (\alpha_1+2\alpha_2) u_{71} v_{71} \big)
-\frac{\alpha_2(\alpha_1+\alpha_2) u_{71} v_{71}}{z^2},
\end{aligned}
$$
and contains no base points.

The energy is given by
$$
\begin{aligned}
&E=-\frac{1 + 2 u_{71} v_{71} + u_{71}^2 v_{71}^3}{u_{71}}-\frac{\alpha_2}{u_{71} v_{71} z}(1 + 2 u_{71} v_{71} + 2 u_{71}^2 v_{71}^3)
-\frac{\alpha_2^2 }{z^2}u_{71} v_{71},
\\
&EJ_{71}=1 + 2 u_{71} v_{71} + u_{71}^2 v_{71}^3+
\frac{\alpha_2}{v_{71} z}(1 + 2 u_{71} v_{71} + 2 u_{71}^2 v_{71}^3)
+\frac{\alpha_2^2 }{z^2}u_{71}^2 v_{71}.
\end{aligned}
$$

The second chart is
\begin{gather*}
u_{72}=u_{42}=\frac{1}{u},
\qquad
v_{72}=\frac{v_{42}-\frac{\alpha_2}{z}}{u_{42}}=\frac{u(uvz-\alpha_2)}{z},
\\
u=\frac{1}{u_{72}},
\qquad
v=\frac{u_{72}}{z}(zu_{72}v_{72}+\alpha_2).
\end{gather*}
In this chart, the exceptional line $\mathcal{L}_7$ is given by equation $u_{72}=0$, while the preimages of $\mathcal{L}_4$, $\mathcal{L}_1$, and $\mathcal{L}_0$ are not visible.

The flow 
$$
\begin{aligned}
u_{72}'=&\frac{u_{72}}{2z}\big(1 + 2 (\alpha_1 + 2 \alpha_2) u_{72}\big) + 1 + 2 u_{72}  + 2 u_{72}^3 v_{72},
 \\
v_{72}'=&-\frac{2 \alpha_2(\alpha_1+\alpha_2) + 4(\alpha_1+2 \alpha_2) u_{72} v_{72} z + 
  v_{72} z (3 + 4 z + 6 u_{72}^2 v_{72} z)}{2 z^2},
\end{aligned}
$$
contains no base points.

\subsubsection{Resolution at $b_7$}
The first chart is

$$
\begin{aligned}
u_{81}=&\frac{u_{51}-\frac{1-\alpha_1-\alpha_2}{8z}}{v_{51}}
\\
=&\frac{u}{8 (u + v)^3 z}\big(
(\alpha_1 + \alpha_2-1) v^2 + (\alpha_1 + \alpha_2 -1 + 4 z)u^2
\\
&\qquad\qquad
+ 2(\alpha_1 + \alpha_2 -1 + 2 z)uv 
 +  8 u  z 
\big)   ,
\\
v_{81}=&v_{51}=\frac{u+v}{u},
\\
u=&\frac{8 z}{v_{81} ( 8 u_{81} v_{81}^2 z-( \alpha_1 + \alpha_2-1) v_{81} - 4 z )},
\\
v=&\frac{8 (v_{81}-1) z}{v_{81} ( 8 u_{81} v_{81}^2 z-( \alpha_1 + \alpha_2-1) v_{81} - 4 z )}.
\end{aligned}
$$
In this chart, the exceptional line $\mathcal{L}_8$ is given by equation $v_{81}=0$, while the preimages of $\mathcal{L}_5$ and $\mathcal{L}_2$ are not visible.

The flow is given by
$$
\begin{aligned}
u_{81}'=&
\left(\frac{1 - 4 \alpha_1 - 4 \alpha_2}{2 z}-10\right)u_{81}
-\frac{\alpha_1 (\alpha_1 + \alpha_2-1)^2}{64 z^3}v_{81}-\frac{\alpha_1 (\alpha_1 + \alpha_2-1)}{16 z^2}
\\
&
+\left(\frac{\alpha_1}{z}-\frac{5 ( \alpha_1 + \alpha_2-1) (\alpha_1 + \alpha_2)}{8 z^2} \right)u_{81}v_{81}
+\frac{(\alpha_1 + \alpha_2-1)^2 (\alpha_1 + \alpha_2)}{32 z^3}
\\
&
+\frac{3 \alpha_1 (\alpha_1 + \alpha_2-1)}{8 z^2} u_{81}v_{81}^2
+\frac{3(\alpha_1+\alpha_2)}{z}u_{81}^2v_{81}^2-\frac{2\alpha_1}{z}u_{81}^2v_{81}^3
\\
&
+
\frac{3}{16z^2(8 u_{81} v_{81}^2 z-(\alpha_1 + \alpha_2-1) v_{81} - 4 z)}
\times
\big({-(\alpha_1 + \alpha_2-1)^3}
\\
&
\qquad\quad
-4 ( \alpha_1 + \alpha_2-1)^2 z
-32  (3(\alpha_1 + \alpha_2-1) + 8 z) z^2u_{81}
\\
&
\qquad\quad
+8  ( \alpha_1 + \alpha_2-1)( \alpha_1 + \alpha_2-1+ 4z)z u_{81} v_{81}
+512 z^3 u_{81}^2 v_{81}
  \big)
,
\\
v_{81}'=&-4 + 4 v_{81} - \frac{\alpha_1+\alpha_2}{z}u_{81}v_{81}^3+\frac{\alpha_1}{z}u_{81}v_{81}^4
 +\frac{(\alpha_1+\alpha_2-1)(\alpha_1+\alpha_2)}{8z^2}v_{81}^2
 \\
 &-\frac{\alpha_1(\alpha_1+\alpha_2-1)}{8z^2}v_{81}^3-\frac{\alpha_1}{2z}v_{81}^2
  +\frac{\alpha_1+\alpha_2}{2z}v_{81}
  \\
  &
+\frac{24z(v_{81}-1)}{8 u_{81} v_{81}^2 z-(\alpha_1 + \alpha_2-1) v_{81} - 4 z }.
\end{aligned}
$$
There are no new base points.

The second chart is
$$
\begin{aligned}
u_{82}=&u_{51}-\frac{1-\alpha_1-\alpha_2}{8z}=\frac{u(u+v+2)}{2(u+v)^2}-\frac{1-\alpha_1-\alpha_2}{8z},
\\
v_{82}=&\frac{v_{51}}{u_{51}-\frac{1-\alpha_1-\alpha_2}{8z}}=
\frac{u+v}{u}\left(\frac{u(u+v+2)}{2(u+v)^2}-\frac{1-\alpha_1-\alpha_2}{8z}\right)^{-1},
\\
u=&\frac{-8 u_{82}^2 z}{v_{82} (4 u_{82} z + v_{82} ( \alpha_1 + \alpha_2 -1 - 8 u_{82} z))},
\\
v=&\frac{8 u_{82} (u_{82} - v_{82}) z}{v_{82} (4 u_{82} z + v_{82} ( \alpha_1 + \alpha_2 -1 - 8 u_{82} z))}.
\end{aligned}
$$
In this chart, the exceptional line $\mathcal{L}_8$ is given by equation $u_{82}=0$, and the preimage of $\mathcal{L}_5$ by $v_{82}=0$.
The preimage of $\mathcal{L}_2$ is not visible.

The Jacobian is
$$
\begin{aligned}
J_{82}=&
\frac{v_{82} (4 u_{82} z + v_{82} (\alpha_1 + \alpha_2-1 - 8 u_{82} z))^3}{512 u_{82}^3 z^3},
\end{aligned}
$$
while the derivative of the Jacobian is
$$
\begin{aligned}
J_{82}'&=\frac{\partial J_{82}}{u_{82}}u_{82}'+\frac{\partial J_{82}}{v_{82}}v_{82}'+\frac{\partial J_{82}}{z}
\\
&=
-\frac{(4 u_{82} z + v_{82} ( \alpha_1 + \alpha_2-1 - 8 u_{82} z))^2}{512 u_{82}^4 z^3}\times
\big(
3 ( \alpha_1 + \alpha_2-1) v_{82}^2 u_{82}'
\\
&\qquad\qquad-
4u_{82}(u_{82} z + v_{82} ( \alpha_1 + \alpha_2 -1- 8 u_{82} z))
v_{82}'
\\
&\qquad\qquad
+3 (\alpha_1 + \alpha_2-1) v_{82}^2 
\big)
.
\end{aligned}
$$

The flow is given by
$$
\begin{aligned}
u_{82}'=&
\frac{8}{v_{82}}-6u_{82}
-\frac{(\alpha_1 + \alpha_2-1)^2}{64 z^3} u_{82} v_{82} (\alpha_1u_{82} v_{82}-2 \alpha_1 -2 \alpha_2)
\\
&
-\frac{1}{4 z}\big(3(1 - \alpha_1 - \alpha_2) +2( 3\alpha_1+3\alpha_2-1) u_{82} - 
   2 \alpha_1 u_{82}^2 v_{82} 
\\
&
\qquad\qquad   - 8 (\alpha_1+\alpha_2) u_{82}^3 v_{82} + 
   4 \alpha_1 u_{82}^4 v_{82}^2\big)
\\
&
+\frac{\alpha_1 + \alpha_2-1}{16 z^2} 
\big(3( \alpha_1 +  \alpha_2-1) - \alpha_1 u_{82} v_{82}
 - 8 (\alpha_1+\alpha_2) u_{82}^2 v_{82}
\\ 
&\qquad\qquad 
 + 4 \alpha_1 u_{82}^3 v_{82}^2\big)
\\
&
+3\frac{32 z^2 + 
   v_{82} ( \alpha_1 + \alpha_2 -1+ 4 z) ( \alpha_1 + \alpha_2-1 -  8 u_{82} z)}
      {4 v_{82} z (8 u_{82}^2 v_{82} z-( \alpha_1 + \alpha_2-1) u_{82} v_{82} - 4 z )},
\\
v_{82}'=&
10v_{82}
+\frac{( \alpha_1 + \alpha_2-1)^2 v_{82}^2 ( \alpha_1 u_{82} v_{82}-2 \alpha_2 -2 \alpha_1)}{64 z^3}
\\
&
+\frac{( \alpha_1 + \alpha_2-1) v_{82}^2 (\alpha_1+10(\alpha_1+ \alpha_2) u_{82} -
    6\alpha_1   u_{82}^2 v_{82})}{16 z^2}
\\
&
+\frac{v_{82}}{2z} \big( 4 \alpha_1 + 4 \alpha_2-1 - 2 \alpha_1 u_{82} v_{82} - 6 (\alpha_1+\alpha_2) u_{82}^2 v_{82} + 4 \alpha_1 u_{82}^3 v_{82}^2\big)
\\
&
+
\frac{3 v_{82}}{16 z^2 ((1 - \alpha_1 - \alpha_2) u_{82} v_{82} - 4 z + 
     8 u_{82}^2 v_{82} z)}\times
\\
&
\quad
\times     
\big(
( \alpha_1 + \alpha_2-1)^3 v_{82}
-4 (\alpha_1 + \alpha_2-1)^2 (2 u_{82}-1) v_{82} z
\\
&
\qquad
-32 ( \alpha_1 + \alpha_2-1) (u_{82} v_{82}-3) z^2
+256(1-2u_{82})z^3
\big).
\end{aligned}
$$
There are no new base points.

The energy is given by
$$
\begin{aligned}
E=&-\frac{128 ( u_{82} v_{82}-1) z^2 (1 - \alpha_1 - \alpha_2 + 8 u_{82} z)}
{u_{82} v_{82} (8 u_{82}^2 v_{82} z-( \alpha_1 + \alpha_2-1) u_{82} v_{82} - 4 z)^3} ,
\\
EJ_{82}=&-\frac{( u_{82} v_{82}-1) (1 - \alpha_1 - \alpha_2 + 8 u_{82} z)}
{4 u_{82}^4 z (8 u_{82}^2 v_{82} z-( \alpha_1 + \alpha_2-1) u_{82} v_{82} - 4 z)^3}
\times
\\
&\qquad\times
\big(4 u_{82} z + 
      v_{82} ( \alpha_1 + \alpha_2 -1 - 
         8 u_{82} z)\big)^3
.
\end{aligned}
$$

\subsubsection{Resolution at $b_8$}

The first chart is 
\begin{gather*}
u_{91}=\frac{u_{62}}{v_{62}+\frac{\alpha_1}{z}}=\frac{z}{v(uvz+\alpha_1)},
\qquad
v_{91}=v_{62}+\frac{\alpha_1}{z}=uv+\frac{\alpha_1}{z},
\\
u=u_{91}v_{91}\left(v_{91}-\frac{\alpha_1}{z}\right),
\qquad
v=\frac{1}{u_{91}v_{91}},
\\
J_{91}=\frac{\partial u_{91}}{\partial u}\frac{\partial v_{91}}{\partial v}-\frac{\partial u_{91}}{\partial v}\frac{\partial v_{91}}{\partial u}
=\frac{z}{v(\alpha_1 + u v z)}=u_{91},
\\
J_{91}'=
-2 u_{91} - 3 u_{91}^2 v_{91}^2 - \frac{\alpha_1(\alpha_1+ \alpha_2) u_{91}^2}{z^2} 
+ \frac{3 u_{91}}{2 z}
+\frac{2(2\alpha_1+\alpha_2)u_{91}^2 v_{91}}{z}.
\end{gather*}
The exceptional line is $\mathcal{L}_9:v_{91}=0$.
In this chart, the preimage of $\mathcal{L}_6$ is given by equation $u_{91}=0$, while the preimages of $\mathcal{L}_3$ and $\mathcal{L}_0$ are not visible.

The flow is 
$$
\begin{aligned}
u_{91}'=&-2 u_{91} - 3 u_{91}^2 v_{91}^2 - \frac{\alpha_1(\alpha_1+ \alpha_2) u_{91}^2}{z^2} 
+ \frac{3 u_{91}}{2 z}
+\frac{2(2\alpha_1+\alpha_2)u_{91}^2 v_{91}}{z},
\\
v_{91}'=&-\frac{1}{u_{91}} + u_{91} v_{91}^3 + \frac{\alpha_1(\alpha_1+\alpha_2) u_{91} v_{91}}{z^2} 
 - \frac{v_{91}}{z} - \frac{(2 \alpha_1+\alpha_2) u_{91} v_{91}^2}{z},
\end{aligned}
$$
and contains no base points.

Energy:
$$
\begin{aligned}
E=&\frac
{(v_{91} z-\alpha_1 ) (\alpha_1 u_{91}^2 v_{91}^2 - z - 2 u_{91} v_{91} z -  u_{91}^2 v_{91}^3 z)}
{u_{91} v_{91} z^2},
\\
EJ_{91}=&-1 - 2 u_{91} v_{91} - u_{91}^2 v_{91}^3 - \frac{\alpha_1^2 u_{91}^2 v_{91}}{z^2} 
+ \frac{ 2 \alpha_1 u_{91}}{z} + \frac{\alpha_1}{v_{91} z} + \frac{2 \alpha_1 u_{91}^2 v_{91}^2}{z}.
\end{aligned}
$$

The second chart is
\begin{gather*}
u_{92}=u_{62}=\frac{1}{v},
\qquad
v_{92}=\frac{v_{62}+\frac{\alpha_1}{z}}{u_{62}}=\frac{v(uvz+\alpha_1)}{z},
\\
u=u_{92}^2v_{92}-\frac{\alpha_1}{z}u_{92},
\qquad
v=\frac{1}{u_{92}}.
\end{gather*}
The exceptional line is $\mathcal{L}_9:u_{92}=0$.
In this chart, the preimages of  $\mathcal{L}_6$, $\mathcal{L}_3$ and $\mathcal{L}_0$ are not visible.

The flow is given by
$$
\begin{aligned}
u_{92}'=&-1 - 2 u_{92} - 2 u_{92}^3 v_{92} + \frac{u_{92}}{2 z} + \frac{(2 \alpha_1+\alpha_2) u_{92}^2}{z}
,
\\
v_{92}'=&2 v_{92} + 3 u_{92}^2 v_{92}^2 + \frac{\alpha_1(\alpha_1+\alpha_2)}{z^2}  -\frac{ 3 v_{92}}{2 z} 
- \frac{ 2(2 \alpha_1+\alpha_2) u_{92} v_{92}}{z}
.
\end{aligned}
$$
and contains no base points.

\end{appendices}

\bibliographystyle{spmpsci} 
\bibliography{reference}   

\end{document}